\documentclass{lmcs}
\pdfoutput=1

\usepackage{lastpage}
\lmcsdoi{19}{4}{22}
\lmcsheading{}{\pageref{LastPage}}{}{}%
{Aug.~30,~2022}{Dec.~08,~2023}{}

\pdfoutput=1


\usepackage{hyperref}
\usepackage{listings}
\usepackage{amssymb}

\theoremstyle{definition}
\newtheorem{theorem}[thm]{Theorem}
\newtheorem{definition}[thm]{Definition}
\newtheorem{problem}[thm]{Problem}
\newtheorem{lemma}[thm]{Lemma}

\newtheorem{corollary}[thm]{Corollary}
\newtheorem{remark}[thm]{Remark}
\newtheorem{example}[thm]{Example}

\usepackage[utf8]{inputenc}
\usepackage{csquotes} %\enquote{} macro

\newcommand{\bbN}{\mathbb{N}}
\newcommand{\bbT}{\mathbb{T}}
\newcommand{\bbS}{\mathbb{S}}
\newcommand{\bbV}{\mathbb{V}}

\newcommand{\calC}{\mathcal{C}}

\newcommand{\calI}{\mathcal{I}}

\newcommand{\MM}{\mathcal{M}} %Minsky machines
\newcommand{\CM}{\mathcal{P}} %one counter machines
\newcommand{\SM}{\mathcal{S}} %stack machines
\newcommand{\inc}[1]{{\texttt{inc}_{#1}}}
\newcommand{\dec}[2]{{\texttt{dec}_{#1}\,{#2}}}

\newcommand{\config}[3]{{#1}{\shortmid}{#2}{\shortmid}{#3}}

\begin{document}

\title[From Halting to Semi-unification]{Constructive Many-one Reduction from the Halting Problem to Semi-unification (Extended Version\rsuper*)}
\titlecomment{{\lsuper*}The present work is an extended version of a prior conference paper~\cite{csl/Dudenhefner22}.}
\thanks{The author is grateful for encouragement and assistance by {Pawe\l} Urzyczyn and
the members of the programming systems lab led by Gert Smolka at Saarland University.
Additionally, the author thanks the anonymous reviewers for insightful remarks on the manuscript and on the intertwined history of semi-unification and polymorphic type inference.}	%optional

\author[A.~Dudenhefner]{Andrej Dudenhefner\lmcsorcid{0000-0003-1104-444X}}

\address{TU Dortmund University, Dortmund, Germany}
\email{andrej.dudenhefner@cs.tu-dortmund.de}

\begin{abstract}
Semi-unification is the combination of first-order unification and first-order matching.
The undecidability of semi-unification has been proven by Kfoury, Tiuryn, and Urzyczyn in the 1990s by Turing reduction from Turing machine immortality (existence of a diverging configuration).
The particular Turing reduction is intricate, uses non-computational principles, and involves various intermediate models of computation.

The present work gives a constructive many-one reduction from the Turing machine halting problem to semi-unification.
This establishes \textsf{RE}-completeness of semi-unification under many-one reductions.
Computability of the reduction function, constructivity of the argument, and correctness of the argument is witnessed by an axiom-free mechanization in the Coq proof assistant.
Arguably, this serves as comprehensive, precise, and surveyable evidence for the result at hand.
The mechanization is incorporated into the existing, well-maintained Coq library of undecidability proofs.
Notably, a variant of Hooper's argument for the undecidability of Turing machine immortality is part of the mechanization.
\end{abstract}

\maketitle

\section{Introduction} 
\label{sec:intro}
\emph{Semi-unification} is the following decision problem: given a finite set of pairs of first-order terms, is there a~substitution $\varphi$ such that for each pair~$(\sigma, \tau)$ in the set of first-order terms we have $\psi(\varphi(\sigma)) = \varphi(\tau)$ for some substitution $\psi$?

Semi-unification was formulated independently by Henglein~\cite{Henglein88}, by Leiß~\cite{leiss1989semi}, and by Kapur, Musser, Narendran and Stillman~\cite{KapurMNS88}.
Both Henglein~\cite{Henglein88} and Leiß~\cite{leiss1989semi} show that semi-unification is equivalent with type inference for functional programming languages with recursive polymorphism\footnote{Mycroft's original type inference algorithm for ML with recursive polymorphism predates the conception of semi-unification and is based on Algorithm~W by Damas and Milner~\cite[Section~6]{DamasM82}.}~\cite{Mycroft84} (for a~detailed account see~\cite{leiss1989polymorphic,Henglein93,KfouryTU93}).
Intuitively, the substitution~$\varphi$ establishes global code invariants (cf.\ first-order unification), and the individual substitutions~$\psi$ establish additional local conditions for each polymorphic function application (cf.\ first-order matching).

While both first-order unification and first-order matching are decidable problems, the status of semi-unification remained open, until answered negatively by Kfoury, Tiuryn, and Urzyczyn~\cite{KfouryTU90,KTU93SemiU}.
The undecidability of semi-unification impacted programming language design and analysis with respect to polymorphic recursion~\cite{LeissH91,jahama1993general}, loop detection~\cite{Purdom87}, and data flow~\cite{FahndrichRD00}.
Another prominent result based on the undecidability of semi-unification is the undecidability of System~F~\cite{girard1972SysF,Reynolds74} typability and type checking~\cite{Wells99,Dud21}.
Of course, the negative result motivated the complementary line of work~\cite{LeissH91} in search for expressive, decidable fragments of semi-unification.
A notable decidable fragment is \emph{acyclic} semi-unification~\cite{MLKfouryTU94}, used for standard ML typability\footnote{\textsf{EXPTIME}-completeness of standard ML typability is shown by Kannellakis and Mitchell~\cite{KanellakisM89} (upper bound) and by Mairson, Kfoury, Tiuryn, and Urzyczyn~\cite{Mairson90, Kfoury90} (lower-bound).}.

Due to the importance of semi-unification in functional programming, it is natural to ask for \emph{surveyable} evidence (both locally and globally in the sense of~\cite{Bassler06}) for its undecidability.
The original undecidability proof~\cite{KfouryTU90} is quite sophisticated, and it was simplified~\cite{Dudenhefner20-SU} and partially mechanized in the Coq proof assistant.
Another peculiarity of semi-unification is that prior to the negative result, it was erroneously claimed to be \emph{decidable}~\cite[U-match~Algorithm]{Purdom87},~\cite[Lemma~5.1]{KfouryTU88},~\cite[{Algorithm~A-1}]{KapurMNS88}.
This serves as further evidence for the intricate nature of semi-unification and calls for surveyable evidence.
Unfortunately, there are several aspects that obstruct surveyability of previous work.

First, existing arguments rely on the undecidability of Turing machine immortality, shown by Hooper~\cite{Hooper66}.
The corresponding construction has received more attention~\cite{KariO08}, however, it was never published in full detail. Hooper remarks:
\begin{displayquote}
\textit{A routine and unimaginative analysis-of-cases proof would point this out more clearly; but it has remained unwritten since, as a~rather tedious insult to the alert, qualified reader, it would surely remain unread.}
\end{displayquote}
While the omissions are justified by accessibility, they hinder verification in full detail.
The existing mechanization~\cite{Dudenhefner20-SU} of the undecidability of semi-unification does not improve upon this aspect, as it assumes the undecidability of Turing machine immortality as an axiom.

Second, existing arguments use non-constructive principles such as excluded middle, König's lemma~\cite{KfouryTU90}, or the fan theorem~\cite{Dudenhefner20-SU} that do not support constructivity of the arguments.
As a~result, anti-classical theories, such as synthetic computability theory~\cite{Bauer06}, may be in conflict with such constructions.
The question arises, whether non-constructive principles are inherent to semi-unification or could be avoided.

Third, existing arguments use Turing reductions and are insufficient to establish \textsf{RE}-completeness of semi-unification under many-one reductions. 
Hitherto, a~many-one reduction from Turing machine halting to semi-unification is not given.

This work improves upon the above aspects as follows.
It provides a~comprehensive chain of many-one reductions from Turing machine halting to semi-unification, replacing immortality with uniform boundedness.
Crucially, each many-one reduction is mechanized in full detail using the Coq proof assistant~\cite{Coq}.
The mechanization witnesses correctness and constructivity of the argument.
Specifically, the notion of a~\emph{constructive} proof is identified with an axiom-free Coq mechanization (cf.\ calculus of inductive constructions).
It neither assumes functional extensionality\footnote{Two functions are equal if their values are equal at every argument (cf.\ homotopy type theory).}, Markov's principle\footnote{If it is impossible that an algorithm does not terminate, then it will terminate (cf.\ Russian constructivism).}, nor the fan theorem\footnote{The topological space of unbounded binary sequences is compact (cf.\ Brouwer's intuitionism).}.
Finally, the mechanization is integrated into the Coq Library of Undecidability Proofs~\cite{CLUP20}, and contributes a~(first of its kind) mechanized variant of Hooper's immortality construction~\cite{Hooper66}.

The described improvements allow for an alternative approach to show many-one equivalence of System F typability and System F type checking, compared to the argument established in the 1990s by Wells~\cite{Wells99}.
The original argument interreduces type checking and typability directly, which requires a~technically sophisticated argument.
Having a~constructive many-one reduction from Turing machine halting to semi-unification at our disposal (together with recursive enumerability of System F typability and type checking), it suffices (and is simpler) to reduce semi-unification to type checking and typability individually~\cite{Dud21}.

\subsection*{Synopsis} The reduction from Turing machine halting to semi-unification is divided into several reduction steps. Each reduction step is many-one and constructive.
That is, a~predicate $P$ over the domain $X$ \emph{constructively many-one} reduces to a~predicate $Q$ over the domain $Y$, if there exists a~computable function $f : X \to Y$ such that for all $x \in X$ we constructively have $P(x) \iff Q(f(x))$.
Additionally, each reduction step is mechanized as part of the Coq Library of Undecidability Proofs~\cite{CLUP20}.

\begin{description}
\item[Section~\ref{sec:MM}] Turing machine halting is reduced to two-counter machine halting (Lemma~\ref{lem:TM_to_CM2}) using Minsky's argument~\cite[Section 14.1]{minsky1967computation}.
This step simplifies the starting machine model.
\item[Section~\ref{sec:CM}] Two-counter machine halting is reduced to one-counter machine $1$-halting (Lemma~\ref{lem:mm-to-cm}), adapting Minsky's observation~\cite[Section 14.2]{minsky1967computation}.
This step prepares the machine model for nested simulation via two-stack machines without symbol search.
\item[Section~\ref{subsec:dlsm}] One-counter machine $1$-halting is reduced to deterministic, length-preserving two-stack machine uniform boundedness (Lemma~\ref{lem:cm1-dlsm}), adapting Hooper's construction for Turing machine immortality~\cite{Hooper66}.
This step transitions the machine problem from halting to uniform boundedness.
\item[Section~\ref{subsec:cssm}] Deterministic, length-preserving two-stack machine uniform boundedness is reduced to  confluent, simple two-stack machine uniform boundedness (Lemma~\ref{lem:dlsm-cssm}), simplifying machine structure.
This step enables reuse of the existing mechanized reduction from a~uniform boundedness problem to semi-unification~\cite{Dudenhefner20-SU}.
\item[Section~\ref{subsec:ssu}] Confluent, simple two-stack machine uniform boundedness is reduced to simple semi-unification (Lemma~\ref{lem:cssm-ssu}), strengthening previous work~\cite{Dudenhefner20-SU}.
This step transitions to an undecidable fragment of semi-unification.
\item[Section~\ref{subsec:r2su}] Simple semi-unification is reduced to right-uniform, two-inequality semi-unifi\-cation (Lemma~\ref{lem:ssu-su2}), establishing the main result (Theorem~\ref{thm:tm-semiu}).
\item[Section~\ref{sec:mec}] Outline of the mechanization of the above reduction steps.
\item[Section~\ref{sec:concl}] Concluding remarks.
\end{description}

\newpage

\section{Two-counter Machines}
\label{sec:MM}
The key insight of recent work~\cite{Dudenhefner20-SU} establishes a~direct correspondence between semi-unification and a~uniform boundedness problem for a~machine model.
In order to reduce Turing machine halting to such a~problem, in this section we consider \emph{two-counter machines} as a~well-understood, mechanized~\cite{ForsterMM2}, and more convenient intermediate model of computation.

Two-counter machines, pioneered by Minsky~\cite[Section 14.1]{minsky1967computation}, constitute a~particularly simple, Turing-complete model of computation.
A two-counter machine (Definition~\ref{def:MM}) stores data in two \emph{counters}, each containing a~natural number.
A program instruction may either increment or decrement a~counter value, and modify the current program index.
The size of a~two-counter machine is the length, denoted $|\cdot|$, of the list, denoted $[ \ldots ]$, of its instructions.
Individual instructions in the list are indexed starting from index $0$.

\begin{definition}[Two-counter Machine ($\MM$)]
\label{def:MM}
A \emph{two-counter machine} $\MM$ is a~list of \emph{instructions} of shape either $\inc{0}$, $\inc{1}$, $\dec{0}{j}$, or $\dec{1}{j}$, where $j \in \bbN$ is a~\emph{program index}.

A \emph{configuration} of $\MM$ is of shape $(i, (a,b))$, where $i \in \bbN$ is the current \emph{program index} and $a,b \in \bbN$ are the current \emph{counter values}.

The \emph{step relation} of $\MM$ on configurations, written ($\longrightarrow_{\MM}$), is given by
\begin{itemize}
\item if $\inc{0}$ is at index $i$ in $\MM$, then $(i, (a,b)) \longrightarrow_{\MM} (i+1, (a+1,b))$;
\item if $\inc{1}$ is at index $i$ in $\MM$, then $(i, (a,b)) \longrightarrow_{\MM} (i+1, (a,b+1))$;
\item if $\dec{0}{j}$ is at index $i$ in $\MM$, then $(i, (0,b)) \longrightarrow_{\MM} (i+1, (0,b))$\\
and $(i, (a+1,b)) \longrightarrow_{\MM} (j, (a,b))$;
\item if $\dec{1}{j}$ is at index $i$ in $\MM$, then $(i, (a,0)) \longrightarrow_{\MM} (i+1, (a,0))$\\
and $(i, (a,b+1)) \longrightarrow_{\MM} (j, (a,b))$;
\item otherwise, we say that the configuration $(i, (a,b))$ is \emph{halting}.
\end{itemize}

The \emph{reachability relation} of $\MM$ on configurations, written ($\longrightarrow_{\MM}^*$), is the reflexive, transitive closure of ($\longrightarrow_{\MM}$).

A configuration $(i,(a,b))$ \emph{terminates} in $\MM$, if we have
$(i, (a,b)) \longrightarrow_{\MM}^* (i', (a', b'))$ for some halting configuration $(i', (a', b'))$.
\end{definition}

\begin{remark}
\label{rem:cm2jnz}
The above Definition~\ref{def:MM} is proposed by Forster and Larchey-Wendling~\cite{ForsterMM2} as a~slight variation of the definition given by Minsky~\cite[Table 11.1-1]{minsky1967computation}.
In particular, the conditional jump on decrement is on a~strictly positive counter value (and not on zero).
Intuitively, conditional jumps on zero (as in Minsky's original definition) severely restrict meaningful control flow when only two counters are available.
In fact, in the setting with only two counters the halting problem for the original definition is \emph{decidable}~\cite[Remark 1]{fscd/Dudenhefner22}, and for universality requires additional control flow instructions.
\end{remark}

Despite its remarkable simplicity, the halting problem for two-counter machines (Problem~\ref{prb:cm2halt}) is undecidable (Corollary~\ref{cor:CM2_undec}).

\begin{problem}[Two-counter Machine Halting]
\label{prb:cm2halt}
Given a~two-counter machine $\MM$ and two natural numbers $a, b \in \bbN$, does the configuration $(0, (a, b))$ terminate in $\MM$?
\end{problem}

\begin{lemma}
\label{lem:TM_to_CM2}
The Turing machine halting problem many-one reduces to the two-counter machine halting problem (Problem~\ref{prb:cm2halt}).
\end{lemma}

\begin{proof}[Proof Sketch]
Minsky describes the simulation of Turing machines by machines with four counters~\cite[Section 11.2]{minsky1967computation} working on Gödel encodings.
Then, machines with four counters are simulated by two-counter machines~\cite[Theorem 14.1-1]{minsky1967computation}.
\end{proof}

\begin{corollary}
\label{cor:CM2_undec}
Two-counter machine halting (Problem~\ref{prb:cm2halt}) is undecidable.
\end{corollary}

\begin{remark}
Minsky's original argument is constructive and, despite some confusion (Remark~\ref{rem:cm2jnz}), is sufficient for our technical result.
However, to avoid ambiguity with respect to the instruction set used, for mechanization we rely on existing work by Forster and Larchey-Wendling~\cite{ForsterMM2}.
This approach emerged as part of the effort to mechanize Hilbert's tenth problem~\cite{H10Coq}.
Specifically, it many-one reduces Turing machine halting via the Post correspondence problem~\cite{post1946PCP,ForsterPCP} to $n$-counter machine halting, and ultimately to two-counter machine halting.
\end{remark}

\begin{remark}
Besides Minsky's original definition, there are several different universal instruction sets for two-counter machines~\cite[Section~2]{Korec96}.
The particular choice in the present work is motivated by the following two facts.
First, the instruction set is sufficiently restricted\footnote{One can combine the $\inc{0}$ and the $\inc{1}$ instructions to one $\inc{\{0,1\}}$ instruction. However, this is insignificant in practice.}.
This allows for shorter proofs based on reductions \emph{from} properties of the particular machine model.
Second, the existing mechanization for the undecidability of the corresponding halting problem is well-maintained, and is readily available in the Coq Library of Undecidability Proofs.
\end{remark}

Commonly, two-counter machines are easily simulated by other machine models, and therefore serve a~key role in undecidability proofs for machine immortality and boundedness problems~\cite{Hooper66,KariO08}.
However, they have one drawback with respect to boundedness properties.
Specifically, there is no natural increasing measure on configurations along the step relation, as it allows for cycles (Example~\ref{xmp:mm-loop}).

\begin{example}
\label{xmp:mm-loop}
Consider the two-counter machine $\MM = [\inc{0}, \dec{0}{0}]$.
For any counter values $a, b \in \bbN$ the configuration~$(0,(a,b))$ does not terminate in $\MM$ because of the cycle
\[(0,(a,b)) \longrightarrow_{\MM} (1,(a+1,b)) \longrightarrow_{\MM} (0,(a,b)) \longrightarrow_{\MM} \ldots\]
Still, the configuration~$(0,(a,b))$ is \emph{bounded}, because from it only a~finite number of distinct configurations is reachable.
In fact, from any configuration at most $2$ distinct configurations are reachable in $\MM$.
\end{example}

One could eliminate cycles by introducing a~third counter, which increases at every step.
This establishes an equivalence between mortality and corresponding boundedness properties.
While inducing a~natural increasing measure (value of the third counter), it incurs additional bookkeeping.
A simulation of such a~three-counter machine by an acyclic two-counter machine is possible~\cite{KariO08}, however, it again obscures the underlying measure.

We address this drawback of two-counter machines with respect to boundedness properties in the following Section~\ref{sec:CM}, building upon Minsky's notion of machines with \emph{one} counter.

\newpage

\section{One-counter Machines}
\label{sec:CM}
As Minsky originally observed~\cite[Section~14.2]{minsky1967computation}, with multiplication and division by constants the halting problem is undecidable for \emph{one-counter machines}.
Specifically, increase and decrease operations for two values $a$ and $b$ can be simulated by multiplication and division respectively by $2$ and $3$ respectively for one value~$2^a 3^b$.

In this section, we further develop Minsky's construction (similarly to~\cite{Wojna99}) of universal machines with one counter in pursuit of two goals.
First, machine runs should be easy to simulate in the stack machine model (Remark~\ref{rem:cm-sim}) used in Section~\ref{sec:sm}.
Second, we require a~measure on machine configurations that increases along the step relation, directly connecting non-termination and unboundedness (Lemma~\ref{lem:cm-ter-bnd}).
Both goals are motivated by the effort necessary to mechanize the overarching result\footnote{
The author firmly believes that proof surveyability does improve with a~simpler mechanization.}.

A program instruction of a~one-counter machine (Definition~\ref{def:CM}), besides modifying the program index, conditionally multiplies the current counter value with either $\frac{2}{1}, \frac{3}{2}, \frac{4}{3}$, or $\frac{5}{4}$.
Notably, such a~multiplication by $\frac{d+1}{d}$ for $d \in \{1, 2, 3, 4\}$ is both easy to simulate uniformly, and strictly increases a~positive counter value.

\begin{definition}[One-counter Machine ($\CM$)]
\label{def:CM}
A \emph{one-counter machine} $\CM$ is a~list of \emph{instructions} of shape $(j, d)$, where $j \in \bbN$ is a~\emph{program index} and $d \in \{1,2,3,4\}$ is a~\emph{counter modifier}.

A \emph{configuration} of $\CM$ is a~pair $(i, c)$, where $i \in \bbN$ is the current \emph{program index} and $c \in \bbN$ such that $c > 0$ is the current \emph{counter value}.

The \emph{step relation} of $\CM$ on configurations, written ($\longrightarrow_{\CM}$), is given by
\begin{itemize}
\item if $|\CM| \leq i$, then $(i, c) \longrightarrow_{\CM} (i, c)$ and we say $(i, c)$ is \emph{halting};
\item if $(j, d)$ is at index $i$ in $\CM$ and $d$ divides $c$,
then $(i, c) \longrightarrow_{\CM} (j, c \cdot \frac{d+1}{d})$;
\item if $(j, d)$ is at index $i$ in $\CM$ and $d$ does not divide~$c$,
then $(i, c) \longrightarrow_{\CM} (i+1, c)$.
\end{itemize}

The \emph{$n$-fold step relation} is denoted $\longrightarrow_{\CM}^n$.
The \emph{reachability relation} of $\CM$ on configurations, written ($\longrightarrow_{\CM}^*$), is the reflexive, transitive closure of ($\longrightarrow_{\CM}$).

A configuration $(i,c)$ \emph{terminates} in $\CM$, if we have $(i, c) \longrightarrow_{\CM}^* (i', c')$ for some halting configuration~$(i', c')$.
\end{definition}

Intuitively, the fractions $\frac{2}{1}, \frac{3}{2}, \frac{4}{3}, \frac{5}{4}$ serve the following goals.
The fractions $\frac{2}{1}$ and $\frac{3}{2}$ multiply by $2$ and $3$ respectively.
The fractions $\frac{3}{2}$ and $\frac{4}{3}$ divide by $2$ and $3$ respectively.
The fraction $\frac{5}{4}$ removes the superfluous factor $4$ introduced by the fraction $\frac{4}{3}$.

\begin{example}
Consider the one-counter machine $\CM=[(1,1), (0,2)]$. The configuration~$(0,1)$ does not terminate in $\CM$ because of the infinite configuration chain
\[(0,1) \longrightarrow_{\CM} (1,1 \cdot \frac{2}{1}) \longrightarrow_{\CM} (0,2 \cdot \frac{3}{2})  \longrightarrow_{\CM} (1,3 \cdot \frac{2}{1})  \longrightarrow_{\CM} (0,6 \cdot \frac{3}{2})  \longrightarrow_{\CM} (1,9 \cdot \frac{2}{1})  \longrightarrow_{\CM} \cdots \]
Naturally, the counter value strictly increases in the above configuration chain.
\end{example}

\begin{remark}
One-counter machines can be understood as a~variant of Conway's \textsc{FRACTRAN} language~\cite{conway1987fractran} with a~relaxed program index transition rule, and restricted to the instruction set ($\frac{2}{1}, \frac{3}{2}, \frac{4}{3}, \frac{5}{4}$).
\end{remark}

The step relation for one-counter machines is total (Lemma~\ref{fac:cm1step}(\ref{fac:cm1steptot})), deterministic (Lemma~\ref{fac:cm1step}(\ref{fac:cm1stepfunct})), and forms increasing chains (Lemma~\ref{fac:cm1step}(\ref{fac:cm1stepfchain}) and Lemma~\ref{fac:cm1step}(\ref{fac:cm1stepfcinc})) up to a~halting configuration, which is necessarily a~trivial configuration loop (Lemma~\ref{fac:cm1step}(\ref{fac:cm1loop})).

\newpage

\begin{lemma}[One-counter Machine Step Relation Properties]~
\label{fac:cm1step}
\begin{enumerate}
\item \label{fac:cm1steptot}
\textbf{Totality:}\\
For all configurations $(i, c)$ there is a~configuration $(i', c')$ such that $(i, c) \longrightarrow_{\CM} (i', c')$.
\item \label{fac:cm1stepfunct}
\textbf{Determinism:}\\
If $(i, c) \longrightarrow_{\CM} (i', c')$ and $(i, c) \longrightarrow_{\CM} (i'', c'')$, then $(i', c') = (i'', c'')$.
\item \label{fac:cm1loop}
\textbf{Halting Property:}\\
A configuration $(i, c)$ is halting iff $(i, c) \longrightarrow_{\CM} (i, c)$.
\item \label{fac:cm1stepfchain}
\textbf{Increasing Measure:}\\
If $(i, c) \longrightarrow_{\CM} (i', c')$ and $(i, c)$ is not halting, then
$|\CM| \cdot c + i < |\CM| \cdot c' + i'$.
\item \label{fac:cm1stepfcinc}
\textbf{Monotone Counter:}
\begin{enumerate}
\item If $(i, c) \longrightarrow_{\CM} (i', c')$, then $c \leq c'$.
\item If $(i, c) \longrightarrow_{\CM}^{|\CM|+1} (i', c')$ and $(i', c')$ is not halting,
then $c < c'$.
\end{enumerate}
\end{enumerate}
\end{lemma}

\begin{proof}
Routine case analysis.
\end{proof}

The above Lemma~\ref{fac:cm1step}(\ref{fac:cm1stepfchain}) gives a~strictly increasing measure $|\CM| \cdot c + i$ on non-halting configurations~$(i, c)$ with respect to the step relation $(\longrightarrow_{\CM})$.
Therefore, any configuration cycle is trivial, i.e.\ the corresponding configuration is halting.
Additionally, by Lemma~\ref{fac:cm1step}(\ref{fac:cm1stepfcinc}), the counter value is guaranteed to increase after $|\CM|+1$ steps, unless a~halting configuration is reached.
This results in a~characterization of termination via boundedness of reachable counter values (Lemma~\ref{lem:cm-ter-bnd}).

\begin{lemma}\label{lem:cm-ter-bnd}
Let $\CM$ be a~one-counter machine.
A configuration $(i,c)$ terminates in $\CM$ iff there is a~$k \in \bbN$ such that for all configurations $(i', c')$ with $(i, c) \longrightarrow_{\CM}^* (i', c')$ we have $c' < k$.
\end{lemma}

\begin{proof}
If from the configuration $(i,c)$ the machine $\CM$ halts after $n$ steps, then $k = 1 + c \cdot 2^n$ bounds the reachable from $(i,c)$ counter values.
Conversely, if $k$ bounds the reachable from the configuration~$(i,c)$ counter values, then after at most $k\cdot (|\CM|+1)$ steps a~halting configuration is reached by Lemma~\ref{fac:cm1step}(\ref{fac:cm1stepfcinc}).
\end{proof}

The halting problem for one-counter machines (Problem~\ref{prb:cm1halt}) starting from the fixed\footnote{We fix a~starting configuration in order to obtain a~simpler recursive simulation in the proof of Lemma~\ref{lem:cm1-dlsm}.} configuration $(0,1)$ is undecidable by reduction from the halting problem for two-counter machines (Lemma~\ref{lem:mm-to-cm}).

\begin{problem}[One-counter Machine $1$-Halting]
\label{prb:cm1halt}
Given a~one-counter machine $\CM$, does the configuration $(0, 1)$ terminate in~$\CM$?
\end{problem}

\begin{lemma}
\label{lem:mm-to-cm}
Two-counter machine halting (Problem~\ref{prb:cm2halt}) many-one reduces to one-counter machine $1$-halting (Problem~\ref{prb:cm1halt}).
\end{lemma}

\begin{proof}
Let $\MM$ be a~two-counter machine and let $a_0, b_0 \in \bbN$ be two starting counter values.

We represent a~pair~$(a, b)$ of counter values of $\MM$ by the family of counter values $2^{a}3^{b}5^{m}$ where $m \in \bbN$.
The purpose of the factor $5^m$ is to increase whenever either the factor~$2^a$ or the factor~$3^b$ decreases.
We simulate instructions of $\MM$ on two counters $(a, b)$ by instructions of $\CM$ on one counter $c = 2^{a}3^{b}5^{m}$ as follows.

\begin{description}
\item[Increase $a$] $2^{a}3^{b}5^{m} \cdot \frac{2}{1} = 2^{a+1}3^{b}5^{m}$;
\item[Increase $b$] $2^{a}3^{b}5^{m} \cdot \frac{2}{1} \cdot \frac{3}{2} = 2^{a}3^{b+1}5^{m}$;
\item[Decrease $a$] $2^{a+1}3^{b}5^{m} \cdot \frac{3}{2} \cdot \frac{4}{3} \cdot \frac{5}{4} = 2^{a}3^{b}5^{m+1}$;
\item[Decrease $b$] $2^{a}3^{b+1}5^{m} \cdot \frac{4}{3} \cdot \frac{5}{4} = 2^{a}3^{b}5^{m+1}$.
\end{description}

\newpage

Simulated decrease instructions are followed by a~simulation of an unconditional jump instruction via $2^{a}3^{b}5^{m} \cdot \frac{2}{1} \cdot \frac{2}{1} \cdot \frac{5}{4} = 2^{a}3^{b}5^{m+1}$, such that on failed decrease the simulation can continue with an appropriate program index.
Initialization of counter values $(a_0, b_0)$ is simulated via $2^0 3^0 5^0 \cdot (\frac{2}{1})^{a_0 + b_0} \cdot (\frac{3}{2})^{b_0} = 2^{a_0} 3^{b_0} 5^0$.

Overall, the configuration $(0, (a_0, b_0))$ terminates in the two-counter machine $\MM$ iff the configuration $(0, 1)$ terminates in the one-counter machine $\CM$.
\end{proof}

\begin{corollary}
One-counter machine $1$-halting (Problem~\ref{prb:cm1halt}) is undecidable.
\end{corollary}

The following Example~\ref{xmp:mm-to-cm} illustrates the construction in the proof of Lemma~\ref{lem:mm-to-cm}, simulating a~looping two-counter machine from Example~\ref{xmp:mm-loop}.

\begin{example}
\label{xmp:mm-to-cm}
Consider the two-counter machine $\MM = [\inc{0}, \dec{0}{0}]$ from Example~\ref{xmp:mm-loop} with the initial counter values $(a_0, b_0) = (1,1)$.
Following the proof of Lemma~\ref{lem:mm-to-cm}, we construct the one-counter machine
\begin{align*}
\CM = [&(1,1), (2,1), (3,2), (4,1), (5,2), (6,3), (3,4)]
\end{align*}
Starting from the configuration $(0,1) = (0, 2^0 3^0 5^0)$ a~run of $\CM$ starts with the initialization
\begin{align*}
(0, 2^0 3^0 5^0) &\stackrel{(1,1)}{\longrightarrow_{\CM}} (1, 2^1 3^0 5^0) \stackrel{(2,1)}{\longrightarrow_{\CM}} (2, 2^2 3^0 5^0)
\stackrel{(3,2)}{\longrightarrow_{\CM}} (3, 2^1 3^1 5^0) = (3, 2^{a_0} 3^{b_0} 5^0)
\end{align*}
Next, on iteration of the infinite loop of $\MM$ is simulated, returning to the program index $3$
\begin{align*}
(3, 2^1 3^1 5^0) &\stackrel{(4,1)}{\longrightarrow_{\CM}} (4, 2^2 3^1 5^0)
\stackrel{(5,2)}{\longrightarrow_{\CM}} (5, 2^1 3^2 5^0)
\stackrel{(6,3)}{\longrightarrow_{\CM}} (6, 2^3 3^1 5^0)
\stackrel{(3,4)}{\longrightarrow_{\CM}} (3, 2^1 3^1 5^1)
\end{align*}
While the factors $2^1$ and $3^1$ remain unchanged after the above simulated iteration, the factor~$5^0$ strictly increases to $5^1$, which is systematic.

As a~result, the configuration $(0,1) = (0, 2^0 3^0 5^0)$ does not terminate in $\CM$, simulating non-termination of the configuration $(0,(a_0, b_0))$ in $\MM$ as follows
\begin{align*}
(0, 2^0 3^0 5^0) &\longrightarrow_{\CM}^3 (3, 2^1 3^1 5^0) \longrightarrow_{\CM}^4 (3, 2^1 3^1 5^1) 
\longrightarrow_{\CM}^4 (3, 2^1 3^1 5^2) \longrightarrow_{\CM}^4 (3, 2^1 3^1 5^3) \longrightarrow_{\CM}^4 \ldots
\end{align*}
Indeed, there is no upper bound on the counter value, in congruence with Lemma~\ref{lem:cm-ter-bnd}.
\end{example}

\begin{remark}
Another universal collection of fractions is $\frac{6}{1}, \frac{1}{2}, \frac{1}{3}$~\cite[Exercise~{14.2-2}]{minsky1967computation}.
However, this collection allows for non-trivial loops and does not induce an increasing measure along the step relation.
\end{remark}

\begin{remark}\label{rem:cm-sim}
It is possible to use the fractions $\frac{6}{1}$, $\frac{5}{2}$, and $\frac{5}{3}$ to establish an increasing measure.
In comparison, the deliberate choice of counter multiplication by $\frac{d+1}{d}$ for a~counter modifier ${d \in \{1,2,3,4\}}$ has several benefits.
First, instructions are of uniform shape.
Therefore, simulation of and reasoning about such instructions requires less case analysis, also impacting the underlying mechanization.
Second, for a~counter value $c = k \cdot d$ multiplication by $\frac{d+1}{d}$ results in $c \cdot \frac{d+1}{d} = c + k$, which is easy to simulate in the two-stack machine model.
\end{remark}

\section{Two-stack Machines}
\label{sec:sm}
In this section, our goal is the simulation of one-counter machines in a~\emph{stack machine} model of computation without unbounded \emph{symbol search}.
Specifically, given a~one-counter machine~$\CM$, we construct a~two-stack machine~$\SM$ such that the configuration $(0, 1)$ terminates in~$\CM$ iff there is a~uniform bound on the number of configurations reachable in $\SM$ from any configuration.

The main difficulty, similar to the undecidability proof for Turing machine immortality~\cite{Hooper66}, is to simulate symbol search (traverse data, searching for a~particular symbol) using a~uniformly bounded machine.
Most problematic are unsuccessful searches that may traverse an arbitrary, i.e.\ not uniformly bounded, amount of data.
The key idea~\cite[Part~IV]{Hooper66} (see also~\cite{KariO08,Jeandel12}) is to implement unbounded symbol search by nested bounded symbol search.

In the present work, we supplement a~high level explanation of the construction (cf.~proof of Lemma~\ref{lem:cm1-dlsm}) with a~comprehensive case analysis as a~mechanized proof in the Coq proof assistant.
This approach, arguably worth striving for in general, has three advantages over existing work.
First, the proof idea is not cluttered with mundane technical details, while the mechanized proof is highly precise.
Second, a~mechanized proof leaves little doubt regarding proof correctness and is open to scrutiny, as there is nothing left to imagination.
Third, the Coq proof assistant tracks any non-constructive assumptions which may hide beneath technical details.

Let us specify the two-stack machine (Definition~\ref{def:sm}) model of computation, which we use to simulate one-counter machines.
An instruction of such a~machine may modify the current machine state, pop from, and push onto two stacks  of binary symbols. We denote the empty word by $\epsilon$, and the concatenation of words $A$ and $B$ by $AB$.

\begin{definition}[Two-stack Machine ($\SM$)]
\label{def:sm}
Let $\bbS$ be a~countably infinite set.
A \emph{two-stack machine}~$\SM$ is a~list of \emph{instructions} of shape $\config{A}{p}{B} \to \config{A'}{q}{B'}$, where $A,B,A',B' \in \{0,1\}^*$ are binary \emph{words} and $p, q \in \bbS$ are \emph{states}.

A \emph{configuration} of $\SM$ is of shape $\config{A}{p}{B}$ where $p \in \bbS$ is the \emph{current state}, $A \in \{0,1\}^*$ is the (reversed) content of the \emph{left stack} and $B \in \{0,1\}^*$ is the content of the \emph{right stack}.

The \emph{step relation} of $\SM$ on configurations, written ($\longrightarrow_{\SM}$), is given by
\begin{itemize}
\item if ${(\config{A}{p}{B} \to \config{A'}{q}{B'}) \in \SM}$, then for ${C, D \in \{0,1\}^*}$ we have 
$\config{CA}{p}{BD} \longrightarrow_{\SM} \config{CA'}{q}{B'D}$.
\end{itemize}

The \emph{reachability relation} of $\SM$ on configurations, written ($\longrightarrow_{\SM}^*$), is the reflexive, transitive closure of ($\longrightarrow_{\SM}$).
\end{definition}

\begin{remark}
A two-stack machine can be understood as a~restricted semi-Thue system on the alphabet $\{0,1\} \cup \bbS$ in which each word contains exactly one symbol from $\bbS$.
Such rewriting systems are employed in the setting of synchronous distributivity~\cite{AnantharamanELNR12}.
\end{remark}

\begin{remark}
To accommodate for arbitrary large machines, the state space~$\bbS$ is not finite.
However, the \emph{effective} state space of any two-stack machine $\SM$ is bounded by the finitely many states occurring in the instructions of~$\SM$ (a list of finite length).
\end{remark}

The key undecidable property of two-stack machines, used in~\cite{Dudenhefner20-SU}, is whether the number of distinct, reachable configurations from any configuration is \emph{uniformly bounded} (Definition~\ref{def:sm-ub}).

\begin{definition}[Uniformly Bounded]
\label{def:sm-ub}
A two-stack machine $\SM$ is \emph{uniformly bounded} if there exists an $n \in \bbN$ such that for any configuration $\config{A}{p}{B}$ we have 
\[|\{ \config{A'}{p'}{B'} \mid \config{A}{p}{B} \longrightarrow_{\SM}^* \config{A'}{p'}{B'} \}| \leq n\]
\end{definition}

Notably, uniform boundedness and \emph{uniform termination}~\cite{MatiyasevichS96} (is every configuration chain finite?) are orthogonal notions, illustrated by the following Examples~\ref{xmp:sm-ub-ut}--\ref{xmp:sm-nub-nut}.

\begin{example}
\label{xmp:sm-ub-ut}
Consider the empty two-stack machine $\SM = []$.
It is uniformly bounded by $n = 1$ and does uniformly terminate as it admits only singleton configuration chains.
\end{example}

\begin{example}
\label{xmp:sm-nub-ut}
Consider the two-stack machine $\SM = [\config{0}{p}{\epsilon} \to \config{\epsilon}{p}{1}]$.
From the configuration~$\config{0^m}{p}{\epsilon}$, where $m \in \bbN$, reachable in $\SM$ configurations are exactly $\config{0^{m-i}}{p}{1^i}$ for $i=0, \ldots, m$.
Therefore, there is no \emph{uniform} bound on the number of reachable configurations.
However, the length of every configuration chain in $\SM$ is finite (bounded by one plus the length of the left stack).
Overall, $\SM$ does uniformly terminate but is not uniformly bounded.
\end{example}

\begin{example}
\label{xmp:sm-ub-nut}
Consider the two-stack machine $\SM = [(\config{0}{p}{\epsilon} \to \config{\epsilon}{q}{1}), (\config{\epsilon}{q}{1} \to \config{0}{p}{\epsilon})]$.
The number of distinct configurations reachable in $\SM$ from any configuration is uniformly bounded by $n = 2$.
However, $\SM$ admits an infinite configuration chain 
\[\config{0}{p}{\epsilon} \longrightarrow_{\SM} \config{\epsilon}{q}{1} \longrightarrow_{\SM} \config{0}{p}{\epsilon} \longrightarrow_{\SM} \config{\epsilon}{q}{1} \longrightarrow_{\SM} \ldots\]
In summary, $\SM$ is uniformly bounded, but does not uniformly terminate.
\end{example}

\begin{example}
\label{xmp:sm-nub-nut}
Consider the two-stack machine $\SM = [\config{0}{p}{\epsilon} \to \config{\epsilon}{p}{1}, \config{\epsilon}{p}{1} \to \config{0}{p}{\epsilon}]$.
Similarly to Example~\ref{xmp:sm-nub-ut}, $\SM$ is not uniformly bounded.
Additionally, similarly to Example~\ref{xmp:sm-ub-nut}, it admits an infinite configuration chain
\[\config{0}{p}{\epsilon} \longrightarrow_{\SM} \config{\epsilon}{p}{1} \longrightarrow_{\SM} \config{0}{p}{\epsilon} \longrightarrow_{\SM} \config{\epsilon}{p}{1} \longrightarrow_{\SM} \ldots\]
In summary, $\SM$ is neither uniformly bounded nor does uniformly terminate.
\end{example}

In literature~\cite{Hooper66,KariO08}, counter machine termination is simulated using uniformly bounded Turing machines directly rather than by two-stack machines.
This is reasonable when omitting technical details regarding the exact Turing machine construction.
However, for verification in full detail, Turing machines are quite unwieldy, compared to two-stack machines.
Unfortunately, we cannot rely on existing mechanized Turing machine programming techniques~\cite{ForsterTM20}, as they establish only functional properties, and are incapable to establish uniform boundedness.

\subsection{Deterministic, Length-preserving Two-stack Machines}
\label{subsec:dlsm}
There are several properties of two-stack machines that are of importance in our construction in order to reuse existing work~\cite{Dudenhefner20-SU}.

For \emph{length-preserving} two-stack machines (Definition~\ref{def:sm-lp}) the sum of lengths of the two stacks is invariant wrt.\ reachability.
For each configuration, length-preservation bounds (albeit, not uniformly) the number of distinct, reachable configurations.
Therefore, reachability is decidable (in polynomial space) for length-preserving two-stack machines.

\begin{definition}[Length-preserving]
\label{def:sm-lp}
A two-stack machine $\SM$ is \emph{length-preserving} if for all instructions ${(\config{A}{p}{B} \to \config{A'}{q}{B'}) \in \SM}$ we have
$0 < |A|+|B| = |A'| + |B'|$.
\end{definition}

\begin{definition}[Deterministic]
\label{def:sm-det}
A two-stack machine $\SM$ is \emph{deterministic} if for all configurations $\config{A}{p}{B}$, $\config{A'}{p'}{B'}$, and $\config{A''}{p''}{B''}$ such that
$\config{A}{p}{B} \longrightarrow_{\SM} \config{A'}{p'}{B'}$ and $\config{A}{p}{B} \longrightarrow_{\SM} \config{A''}{p''}{B''}$ we have $\config{A'}{p'}{B'} = \config{A''}{p''}{B''}$.
\end{definition}

\begin{example}
\label{xmp:sm-nondet}
Two-stack machines in Examples~\ref{xmp:sm-ub-ut}--\ref{xmp:sm-ub-nut} are deterministic and length-preserving.
The two-stack machine $\SM = [\config{0}{p}{\epsilon} \to \config{\epsilon}{p}{1}, \config{\epsilon}{p}{1} \to \config{0}{p}{\epsilon}]$ from Example~\ref{xmp:sm-nub-nut} is length-preserving, but not deterministic because of $\config{0}{p}{1} \longrightarrow_{\SM} \config{\epsilon}{p}{11}$ and $\config{0}{p}{1} \longrightarrow_{\SM} \config{00}{p}{\epsilon}$.
\end{example}

\begin{remark}
Deterministic, length-preserving two-stack machines can be understood as a~generalization of intercell Turing machines~\cite[Section 3]{KTU93SemiU}, for which the transition rule has a~bounded read/write/move radius around the current head position.
Tape content on the left (resp. right) of the current head position is exactly the content of the left (resp. right) stack.
\end{remark}

The key undecidable problem, that in our argument assumes the role of Turing machine immortality of previous approaches~\cite{KTU93SemiU,Dudenhefner20-SU}, is uniform boundedness of deterministic, length-preserving two-stack machines (Problem~\ref{prb:dlsmub}).
A central insight of the present work is that using this problem as an intermediate step we neither require Turing reductions, König's lemma (cf.~\cite{KTU93SemiU}), nor the fan theorem (cf.~\cite{Dudenhefner20-SU}).
Additionally, a~mechanization for both the reduction from counter machine halting to this problem and the reduction from this problem to semi-unification is of manageable size.

\begin{problem}[Deterministic, Length-preserving Two-stack Machine Uniform Boundedness]
\label{prb:dlsmub}
Given a~deterministic, length-preserving two-stack machine $\SM$, is $\SM$ uniformly bounded?
\end{problem}

The original undecidability proof of semi-unification contains a~hint \cite[Proof of Corollary 6]{KTU93SemiU} that Turing machine immortality may be avoided.
Accordingly, Lemma~\ref{lem:cm1-dlsm} captures the decisive step, which avoids Turing machine immortality.

\begin{remark}[Naive Simulation]
\label{rem:cm-sim-naive}
Given a~one-counter machine $\CM$, we can construct a~length-preserving two-stack machine $\SM$ such that $(0,1)$ terminates in $\CM$ iff configurations $\config{A1}{0}{01B}$ for $A, B \in \{0, 1\}^*$ are uniformly bounded in~$\SM$.

A $\CM$-configuration $(i,c)$ can be represented by the $\SM$-configuration $\config{1}{i}{0^c10^m}$, where $m\in \bbN$, such that $0^m$ is long enough to simulate a~terminating run.
In particular, the counter value~$c$ is unary encoded as $0^c1$ with $1$ as separator.
A $\CM$-instruction~$(j,d)$ at index $i$ is simulated as follows.
The $\SM$-instruction $\config{\epsilon}{i_?}{0^d} \longrightarrow_\SM \config{0^d}{i_?}{\epsilon}$ tests for divisibility by $d$, moving consecutive blocks of $d$ zeroes from the right stack to the left stack.
The divisibility test either fails or succeeds.
\begin{description}
\item[Failure $(i,c) \longrightarrow_\CM (i+1,c)$] The $\SM$-instruction $\config{\epsilon}{i_?}{0^k1} \longrightarrow_\SM \config{\epsilon}{i_\#}{0^k1}$ can be applied, where $k$ is the remainder such that $0 < k < d$.
In this case, we move the zeroes back to the right stack, and via the $\SM$-instruction $\config{1}{i_\#}{\epsilon} \longrightarrow_\SM \config{1}{i+1}{\epsilon}$ we reach the $\SM$-configuration $\config{1}{i+1}{0^c10^m}$, representing the next $\CM$-configuration $(i+1,c)$.
\item[Success $(i,c) \longrightarrow_\CM (j,c \cdot \frac{d+1}{d})$] The $\SM$-instruction $\config{\epsilon}{i_?}{1} \longrightarrow_\SM \config{\epsilon}{i_!}{1}$ can be applied.
In this case, multiplication by $\frac{d+1}{d}$ is simulated as follows.
For each block of consecutive~$d$ zeroes on the left stack shift the separator $1$ on the right stack one position to the right (Remark~\ref{rem:cm-sim}).
Finally, via the $\SM$-instruction $\config{1}{i_!}{\epsilon} \longrightarrow_\SM \config{1}{j}{\epsilon}$, we arrive at the $\SM$-configuration $\config{1}{j}{0^{\frac{c(d+1)}{d}}10^{m-\frac{c}{d}}}$, representing the next $\CM$-configuration $(j,c \cdot \frac{d+1}{d})$.
\end{description}
We obtain $(i,c) \longrightarrow_\CM^* (i',c')$ iff $\config{A1}{i}{0^c10^{c'-c}B} \longrightarrow_\SM^* \config{A1}{i'}{0^{c'}1B}$ for all words~$A,B \in \{0,1\}^*$.
Therefore, the configuration $(0,1)$ terminates in $\CM$ iff configurations $\config{A1}{0}{01B}$ are uniformly bounded in~$\SM$ (the uniform bound is derived from the halting counter value).
\end{remark}

Unfortunately, the naive construction in the above Remark~\ref{rem:cm-sim-naive} in general does not describe a~uniformly bounded two-stack machine.
Most importantly, the naive construction assumes and preserves the \enquote{safe} format $\config{A1}{i}{0^c1B}$ of machine configurations $(i, c)$, where $c$ is less or equal to the halting counter value.
Arbitrary configurations do not need to follow this format.
At fault is a~symbol search (for the symbol $1$ on the right stack) that needs to traverse an arbitrary amount of data, illustrated by the following Example~\ref{xmp:cm-sim-naive}.

\begin{example}
\label{xmp:cm-sim-naive}
Consider a~one-counter machine $\CM$ containing the instruction $(j, d)$ at index $i$.
The naive construction of a~two-stack machine $\SM$ from Remark~\ref{rem:cm-sim-naive} contains the $\SM$-instruction $\config{\epsilon}{i_?}{0^d} \longrightarrow_\SM \config{0^d}{i_?}{\epsilon}$.
For any bound $n \in \bbN$ we have that $|\{ \config{A}{p}{B} \mid \config{1}{i_?}{0^{n \cdot d}1} \longrightarrow_\SM^* \config{A}{p}{B}\}| > n$.
Therefore, $\SM$ is not uniformly bounded.
\end{example}

\newpage

The ingenious idea by Hooper~\cite[Part~IV]{Hooper66} is to use \emph{nested simulation} for symbol search.
In the present scenario, to search for the symbol $1$ on the right stack, start a~new simulation from the $\CM$-configuration $(0,1)$ inside the space of consecutive zeroes on the right stack.
Nested simulation achieves three goals.
First, it transitions into a~\enquote{safe} configuration format $\config{A1}{i}{0^c1B}$, where $c$ is less or equal to the halting counter value.
Second, by inspecting symbols in the immediate neighborhood of the separator $1$ on the right stack, it checks whether the original search for the symbol~$1$ succeeds in appropriately limited space.
Third, in case the space limit for symbol search is exceeded, it simulates a~terminating run of the given machine.

Similarly to~\cite[Theorem 7]{KariO08}, we adapt Hooper's nested simulation for symbol search in the following Lemma~\ref{lem:cm1-dlsm}.

\begin{lemma}
\label{lem:cm1-dlsm}
One-counter machine $1$-halting (Problem~\ref{prb:cm1halt}) many-one reduces to deterministic, length-preserving two-stack machine uniform boundedness (Problem~\ref{prb:dlsmub}).
\end{lemma}

\begin{proof}
Let $\CM$ be a~one-counter machine.
We follow the naive construction in Remark~\ref{rem:cm-sim-naive} with the following difference.
In order to search for the symbol $1$ on the right stack in some configuration $\config{A}{p}{0^{k+3}B}$, transition into the configuration $\config{AC1}{0}{01B}$, i.e.\ reset the program index $p$ to $0$ and retain the binary encoding of $p$ in $C$ of fixed length~$k$.

Assume that the $\CM$-configuration $(0,1)$ terminates in $\CM$ and let $c$ be the counter value of the halting configuration.
For a~nested simulation from the configuration $\config{AC1}{0}{01B}$ in search of symbol~$1$ in $B$, there are three cases.
\begin{description}
\item[Case $B = 0^m1D$, where $m < c-1$ and ${D \in \{0,1\}^*}$] The number $m$ of zeroes on the right stack is too small to accommodate for $c$.
Eventually, the nested simulation is unable to simulate counter increase.
This is detected by inspecting symbols in the immediate neighborhood of the separator $1$ on the right stack.
In this situation, the initial search for the symbol~$1$ in $B$ succeeds, the previous level configuration is restored, and control is returned to the previous level.
The content of $D$ is not accessed, and symbol search succeeds in bounded space (the bound is derived from $m < c-1$).
\item[Case $B = 0^m$, where $m < c-1$] The size of the right stack is too small to accommodate for $c$.
Eventually, the nested simulation is unable to shift the separator $1$ to the right (in fact, apply any instruction), and halts.
In this situation, the initial search for the symbol~$1$ fails in bounded space (the bound is derived from $m < c-1$).

Since the original configuration $\config{A}{p}{0^{k+3}B}$ is not in the \enquote{safe} format (missing separator~$1$ on the right stack), this case handles ill-formed configurations in bounded space, and the search is immaterial.
\item[Case $B = 0^{c-1}D$, where $D \in \{0,1\}^*$] There are enough consecutive zeroes in~$B$ to simulate a~terminating run of $\CM$ from the initial configuration $(0,1)$.
The content of $D$ is not accessed, and the simulation terminates in bounded space (the bound is derived from $c-1$).

Since the original configuration $\config{A}{p}{0^{k+3}B}$ is not in the \enquote{safe} format (too many consecutive zeroes on the right stack), this case handles ill-formed configurations in bounded space, and the search is immaterial.
\end{description}
Each case may require further symbol search, and therefore introduce further nested simulation.
Notably, consecutive nested simulation is performed inside the space of at most $c$ consecutive zeroes (followed by the separator $1$ in the \enquote{safe} format).
Therefore, the nesting depth is at most $c$.
In each case (using Lemma~\ref{fac:cm1step}) a~bound on the number of configurations for the nested simulation can be derived from~$c$.
Therefore, $\SM$ is uniformly bounded.

For the converse, assume that the $\CM$-configuration $(0,1)$ does not terminate in $\CM$.
By Lemma~\ref{lem:cm-ter-bnd}, there is no bound on reachable counter values.
Therefore, there is no uniform bound on the number of configurations reachable from configurations $\config{1}{0}{010^m}$, where $m \in \bbN$ is arbitrary large.
\end{proof}

\begin{corollary}
Deterministic, length-preserving two-stack machine uniform boundedness (Problem~\ref{prb:dlsmub}) is undecidable.
\end{corollary}

\begin{remark}
At first glance, recursive nested simulation in the proof of Lemma~\ref{lem:cm1-dlsm} seems superfluous.
After all, a~nested simulation establishes a~\enquote{safe} configuration format with bounded (naive) symbol search.
However, it is not possible to remember reliably in the current configuration (e.g.\ in the current state) whether a~nested simulation is running.
An ill-formed configuration may be locally indistinguishable from a~configuration of a~nested simulation.
Therefore, any (recursive) symbol search is performed using nested simulation.
\end{remark}

\begin{remark}
The exact analysis of the nested simulation in the proof of Lemma~\ref{lem:cm1-dlsm} requires a~tremendous inductive proof with many corner cases.
Arguably, it is unreasonable for a~human without mechanical assistance to write it down in full detail (cf.\ Hooper's remark in Section~\ref{sec:intro}).
Additionally, it would require a~comparable amount of effort for others to verify such a~massive construction.
This is why, in order to guarantee its correctness, a~mechanized proof of Lemma~\ref{lem:cm1-dlsm} is adequate (Section~\ref{sec:mec}) to establish the result.
\end{remark}

\subsection{Confluent, Simple Two-stack Machines}
\label{subsec:cssm}
To further simplify the construction, we consider confluent (instead of deterministic), simple (Definition~\ref{def:sm-simp}, cf.~\cite[Definition~16]{Dudenhefner20-SU}) two-stack machines.
Uniform boundedness for such two-stack machines (Problem~\ref{prb:cssmub}) is well-suited for reduction to a~fragment of semi-unification (cf.\ Section~\ref{subsec:ssu}).

\begin{definition}[Simple]
\label{def:sm-simp}
A two-stack machine $\SM$ is \emph{simple}
if for all instructions\\
$(\config{A}{p}{B} \to \config{A'}{q}{B'}) \in \SM$ we have
$1 = |A|+|B| = |A'| + |B'| = |A| + |A'| = |B| + |B'|$.
\end{definition}

\begin{remark}
A deterministic, simple two-stack machine is just another way to present a~deterministic Turing machine.
The left and right stacks contain the respective tape content to the left and to the right from the current head.
Reading and writing at the head while moving the head position is easily presented as simple instructions (cf.~\cite[Remark 19]{Dudenhefner20-SU}).
\end{remark}

\begin{remark}
Turing machine immortality is reducible to uniform boundedness of deterministic, simple two-stack machines by a~bounded Turing reduction~\cite[Theorem 2]{Dudenhefner20-SU}.
However, the argument uses the fan theorem, therefore, it is crucial for the present argument to not rely on this particular reduction.
\end{remark}

\begin{definition}[Confluent]
\label{def:sm-conf}
A two-stack machine $\SM$ is \emph{confluent} if for all configurations
$\config{A}{p}{B}$, $\config{A'}{p'}{B'}$, and $\config{A''}{p''}{B''}$ such that
$\config{A}{p}{B} \longrightarrow_{\SM}^* \config{A'}{p'}{B'}$ and $\config{A}{p}{B} \longrightarrow_{\SM}^* \config{A''}{p''}{B''}$ there exists a~configuration $\config{C}{q}{D}$ such that
$\config{A'}{p'}{B'} \longrightarrow_{\SM}^* \config{C}{q}{D}$ and $\config{A''}{p''}{B''} \longrightarrow_{\SM}^* \config{C}{q}{D}$.
\end{definition}

Any deterministic two-stack machine is confluent (Lemma~\ref{lem:det-conf}), but not necessarily vice versa (Example~\ref{xmp:sm-nondet-conf}).

\begin{lemma}
\label{lem:det-conf}
If a~two-stack machine $\SM$ is deterministic, then $\SM$ is confluent.
\end{lemma}

\begin{proof}
For a~deterministic two-stack machine $\SM$ we have that $\config{A}{p}{B} \longrightarrow_{\SM}^* \config{A'}{p'}{B'}$ and $\config{A}{p}{B} \longrightarrow_{\SM}^* \config{A''}{p''}{B''}$ implies that $\config{A'}{p'}{B'} \longrightarrow_{\SM}^* \config{A''}{p''}{B''}$ or $\config{A''}{p''}{B''} \longrightarrow_{\SM}^* \config{A'}{p'}{B'}$.
\end{proof}

\newpage

\begin{example}
\label{xmp:sm-nondet-conf}
The two-stack machine $\SM = [\config{0}{p}{\epsilon} \to \config{\epsilon}{p}{1}, \config{\epsilon}{p}{1} \to \config{0}{p}{\epsilon}]$ from Example~\ref{xmp:sm-nub-nut} is non-deterministic, as observed in Example~\ref{xmp:sm-nondet}. However, $\SM$ is confluent.
The instruction $\config{0}{p}{\epsilon} \to \config{\epsilon}{p}{1}$ can be \enquote{undone} by the instruction $\config{\epsilon}{p}{1} \to \config{0}{p}{\epsilon}$ and vice versa.
Therefore, configuration chains can be reversed to join branching computation.
This technique is used in the proof of Lemma~\ref{lem:dlsm-cssm}.
\end{example}

Compared to deterministic machines, confluent machines are quite practical. For example, a~confluent machine may (without additional bookkeeping) \enquote{try out} different configuration chains before choosing the preferable one (cf.\ proof of Lemma~\ref{lem:dlsm-cssm}).

\begin{problem}[Confluent, Simple Two-stack Machine Uniform Boundedness]
\label{prb:cssmub}
Given a~confluent, simple two-stack machine~$\SM$, is $\SM$ uniformly bounded?
\end{problem}

By Lemma~\ref{lem:det-conf}, the above Problem~\ref{prb:cssmub} subsumes deterministic, simple two-stack machine uniform boundedness~\cite[Problem~26]{Dudenhefner20-SU}.
This, in combination with the following Lemma~\ref{lem:dlsm-cssm}, allows for adaptation of previous work\footnote{It is possible to carry out the construction in the original, deterministic scenario without adaptation. However, this is technically more challenging and provides no tangible benefit.} in Section~\ref{subsec:ssu}.

\begin{lemma}
\label{lem:dlsm-cssm}
Deterministic, length-preserving two-stack machine uniform boundedness (Problem~\ref{prb:dlsmub}) many-one reduces to confluent, simple two-stack machine uniform boundedness (Problem~\ref{prb:cssmub}).
\end{lemma}

\begin{proof}[Proof]
Our objective is to shorten length-preserving instructions, while maintaining confluence. This is routine, storing local stack information in additional fresh states.
For example, an instruction $\config{00}{p}{\epsilon} \longrightarrow \config{11}{q}{\epsilon}$ can be replaced by the simple instructions $\config{0}{p}{\epsilon} \longrightarrow \config{\epsilon}{p_1}{0}$, $\config{0}{p_1}{\epsilon} \longrightarrow \config{\epsilon}{p_2}{0}$, $\config{\epsilon}{p_2}{0} \longrightarrow \config{1}{p_3}{\epsilon}$, and $\config{\epsilon}{p_3}{0} \longrightarrow \config{1}{q}{\epsilon}$, where $p_1, p_2, p_3$ are fresh states.
This results in the configuration chain 
\[\config{00}{p}{\epsilon} \longrightarrow \config{0}{p_1}{0} \longrightarrow \config{\epsilon}{p_2}{00} \longrightarrow \config{1}{p_3}{0} \longrightarrow \config{11}{q}{\epsilon}\]
which simulates the instruction $\config{00}{p}{\epsilon} \longrightarrow \config{11}{q}{\epsilon}$.

Notably, it is difficult to maintain determinism for failed look-ahead, when we want to preserve elegance of the above simulation.
However, in order to maintain confluence it suffices to add reverse instructions from fresh states, that is $\config{\epsilon}{p_1}{0} \longrightarrow \config{0}{p}{\epsilon}$, $\config{\epsilon}{p_2}{0} \longrightarrow \config{0}{p_1}{\epsilon}$, and ${\config{1}{p_3}{\epsilon} \longrightarrow \config{\epsilon}{p_2}{0}}$. Therefore, any failed attempt to read local stack information can be reversed (cf.~Example~\ref{xmp:sm-nondet-conf}) and computation is confluent.
\end{proof}

\begin{corollary}
Confluent, simple two-stack machine uniform boundedness (Problem~\ref{prb:cssmub}) is undecidable.
\end{corollary}

Finally, let us justify the detour from one-counter machines via deterministic, length-preserving two-stack machines to confluent, simple two-stack machines.
By far the most complicated part of the overall reduction is the simulation of a~terminating machine by a~uniformly bounded machine (cf.~Lemma~\ref{lem:cm1-dlsm}).
For this reduction, it is beneficial to have a~restricted source machine model and an expressive target machine model.
Of course, a~\emph{direct} simulation of terminating one-counter machines by uniformly bounded, confluent, simple two-stack machines is possible.
However, the restricted nature of simple instructions would induce an unwieldy two-stack machine construction, rendering the mechanization more laborious.

\newpage

\section{Semi-unification}
Semi-unification (Problem~\ref{prb:semiu}) can be understood as a~combination of first-order unification (cf.\ substitution~$\varphi$) and first-order matching (cf.\ substitutions~$\psi$).
For the undecidability of semi-unification~\cite[Theorem~12]{KfouryTU93}, it suffices to restrict the syntax of the underlying terms (Definition~\ref{def:term}) to variables together with a~single binary constructor $(\to)$.

In this section, we recapitulate necessary definitions and properties of semi-unification from existing work~\cite{KTU93SemiU,Dudenhefner20-SU}, in order to complete a~constructive many-one reduction from Turing machine halting to semi-unification (Theorem~\ref{thm:tm-semiu}).

\begin{definition}[Terms $(\bbT)$] 
\label{def:term}
Let $\alpha, \beta, \gamma$ range over a~countably infinite set $\bbV$ of \emph{variables}.
The set of \emph{terms} $\bbT$, ranged over by $\sigma, \tau$, is given by the following grammar
\[\sigma, \tau \in \bbT ::= \alpha \mid \sigma \to \tau\]
\end{definition}

\begin{definition}[Substitution $(\varphi), (\psi)$]
\label{def:substitution}
A \emph{substitution} $\varphi : \bbV \to \bbT$ assigns terms to variables, and is tacitly lifted to terms.
\end{definition}

\begin{problem}[Semi-unification~{\cite[\textsf{SUP}]{KTU93SemiU},~\cite[Problem~3]{Dudenhefner20-SU}}]
\label{prb:semiu}~\\
Given a~set $\calI = \{\sigma_1 \leq \tau_1, \ldots, \sigma_n \leq \tau_n\}$ of \emph{inequalities},
is there a~substitution $\varphi$ such that for each inequality $(\sigma \leq \tau) \in \calI$
there exists a~substitution $\psi : \bbV \to \bbT$ such that $\psi(\varphi(\sigma)) = \varphi (\tau)$?
\end{problem}

\begin{remark}
\label{rem:semiu-rec}
As given by Definition~\ref{def:substitution}, the set of substitutions is not countable.
However, in any semi-unification instance~$\calI$ the number of inequalities (consisting of first-order terms) is finite.
Therefore, restricting substitutions to be finite maps (from the relevant variables) does not change the expressive power of semi-unification.
As a~result, semi-unification is recursively enumerable.
\end{remark}

The following Example~\ref{xmp:supos} and Example~\ref{xmp:suneg} illustrate a~positive and negative instances of semi-unification.

\begin{example}
\label{xmp:supos}
Consider $\calI = {\{ \alpha \leq \alpha \to \alpha,\ \alpha \leq \alpha \to \alpha \to \alpha \}}$. 
The semi-unification instance~$\calI$ is solved by the substitution $\varphi$ such that $\varphi(\alpha) = \alpha$.
\begin{itemize}
\item
For the inequality $\alpha \leq \alpha \to \alpha$ there exists a~substitution~$\psi_1$ such that $\psi_1(\alpha) = \alpha \to \alpha$.
We have $\psi_1(\varphi(\alpha)) = \varphi(\alpha \to \alpha)$.
\item For the inequality $\alpha \leq \alpha \to \alpha \to \alpha$ there exists a~(different from $\psi_1$) substitution~$\psi_2$ such that $\psi_2(\alpha) = \alpha \to \alpha \to \alpha$.
We have $\psi_2(\varphi(\alpha)) = \varphi(\alpha \to \alpha \to \alpha)$.
\end{itemize}
\end{example}

\begin{example}
\label{xmp:suneg}
Consider $\calI = {\{ \alpha \to \alpha \leq \alpha \}}$. 
The semi-unification instance $\calI$ has no solution.
Assume that there exist substitutions $\varphi$ and $\psi$ such that $\psi(\varphi(\alpha \to \alpha)) = \varphi(\alpha)$.
Therefore, the size of the syntax tree of $\varphi(\alpha)$ is twice the size of the syntax tree $\psi(\varphi(\alpha))$ which is not possible for (non-empty, finite) terms.
\end{example}

The following Example~\ref{xmp:su-rec} compares type inference for \emph{parametric polymorphism}~\cite{Milner78} based on unification, and for \emph{recursive polymorphism}~\cite{Mycroft84} based on semi-unification.

\begin{example}
\label{xmp:su-rec}
Consider the following functional program \lstinline|f|, where \lstinline|not| is Boolean negation
\begin{lstlisting}
f b x = if b then x else f (f (not b) (not b)) x
\end{lstlisting}
Basically, \lstinline|f b x| reduces to \lstinline|x|, regardless of the value of \lstinline|b|.
However, both \lstinline|x| and \lstinline|not b| are used as second argument for some recursive call of \lstinline|f|.

\newpage

For languages with \emph{parametric polymorphism}, recursive calls are considered \emph{monomorphic}.
Therefore, type inference for \lstinline|f| unifies the type of \lstinline|x| and the type of \lstinline|not b|.
This results in the inferred type $\mathtt{f} : \mathtt{bool} \to \mathtt{bool} \to \mathtt{bool}$.

For languages with \emph{recursive polymorphism} the type of each individual recursive call can be instantiated separately.
The problematic recursive call \lstinline|f (not b) (not b)| is associated with the inequality $\sigma \leq \tau$, where $\sigma = \mathtt{bool} \to \alpha \to \alpha$ is the global type scheme and $\tau = \mathtt{bool} \to \mathtt{bool} \to \mathtt{bool}$ is the local type scheme.
Any substitution $\varphi$ such that $\varphi(\alpha) = \beta$ for some $\beta \in \bbV$ solves the inequality $\sigma \leq \tau$ setting $\psi(\beta) = \mathtt{bool}$.
That is, we have $\psi(\varphi(\sigma)) = \varphi(\tau)$.
This results in the inferred type $\mathtt{f} : \mathtt{bool} \to \alpha \to \alpha$.
\end{example}

Unfortunately, semi-unification does not admit a~decision procedure based on an occurs-check, which is a~common approach to both first-order unification and first-order matching~\cite{BaaderS01}.
However, it is challenging to construct an unsolvable example of manageable size, for which the occurs-check is not triggered.

Originally~\cite[Theorem 12]{KTU93SemiU}, semi-unification is proven undecidable by Turing reduction from Turing machine immortality~\cite{Hooper66}. As intermediate problems, the argument relies on symmetric intercell Turing machine boundedness, path equation derivability, and termination of a~custom-tailored redex contraction procedure.
Additionally, the argument uses König’s lemma and it is not obvious whether it can be presented constructively.

A modern approach~\cite[Theorem 4]{Dudenhefner20-SU} simplifies the original argument.
It still relies on a~Turing reduction from Turing machine immortality, but uses only deterministic, simple two-stack machine uniform boundedness to show undecidability of a~fragment of semi-unification.
Additionally, it relies on the fan theorem, which is strictly weaker than König’s lemma and is valid in Brouwer's intuitionism.
The argument is partially mechanized in Coq.

In the remainder of this section we briefly recapitulate and reuse the modern approach~\cite{Dudenhefner20-SU} in the more general case of confluent, simple two-stack machines.
This allows us to avoid Turing machine immortality, Turing reductions, and the fan theorem in the overall argument.

\subsection{Simple Semi-unification}
\label{subsec:ssu}
In this section, we recapitulate the intermediate problem of \emph{simple semi-unification} (Problem~\ref{prb:ssu})~\cite[Problem 15]{Dudenhefner20-SU}, which connects stack machine computation and semi-unification.
Intuitively, term variables represent machine states, simple constraints (Definition~\ref{def:sconstr}) represent local stack transformations, and the model relation (Definition~\ref{def:smodel}) captures machine reachability via substitutions.

\begin{definition}[Simple Constraint~{\cite[Definition~6]{Dudenhefner20-SU}}]
\label{def:sconstr}
A \emph{simple constraint} has the shape $\config{a}{\alpha}{\epsilon} \doteq \config{\epsilon}{\beta}{b}$, where $a,b \in \{0,1\}$ are symbols and $\alpha, \beta \in \bbV$ are variables.
\end{definition}

\begin{definition}[Model Relation~{\cite[Definition~9]{Dudenhefner20-SU}}]
\label{def:smodel}
A substitution triple $(\varphi, \psi_0, \psi_1)$ \emph{models} a~simple constraint $\config{a}{\alpha}{\epsilon} \doteq \config{\epsilon}{\beta}{b}$, written ${(\varphi, \psi_0, \psi_1) \models \config{a}{\alpha}{\epsilon} \doteq \config{\epsilon}{\beta}{b}}$, if one of the following conditions holds
\begin{itemize}
\item $b = 0$ and $\psi_a(\varphi(\alpha)) \to \tau = \varphi(\beta)$ for some term $\tau \in \bbT$;
\item $b = 1$ and $\sigma \to \psi_a(\varphi(\alpha)) = \varphi(\beta)$ for some term $\sigma \in \bbT$.
\end{itemize}
\end{definition}

The intuition behind a~simple constraint $\config{a}{\alpha}{\epsilon} \doteq \config{\epsilon}{\beta}{b}$ is as follows.
The symbol $a \in \{0,1\}$ specifies the associated substitution $\psi_a$.
This corresponds to considering two inequalities~\cite[Remark after Theorem~12]{KfouryTU93}.
The symbol $b \in \{0,1\}$ specifies whether we are interested in the left of the right subterm of $\varphi(\beta)$.
This allows for arbitrary deep exploration of term structure.

\newpage

\begin{problem}[Simple Semi-unification {\cite[Problem~15]{Dudenhefner20-SU}}]
\label{prb:ssu}
Given a~finite set $\calC$ of simple constraints, do there exist substitutions $\varphi, \psi_0, \psi_1 : \bbV \to \bbT$ such that for each simple constraint $(\config{a}{\alpha}{\epsilon} \doteq \config{\epsilon}{\beta}{b}) \in \calC$ we have $(\varphi, \psi_0, \psi_1) \models \config{a}{\alpha}{\epsilon} \doteq \config{\epsilon}{\beta}{b}$?
\end{problem}

The following Examples~\ref{xmp:ssu-yes}--\ref{xmp:ssu-no} illustrate a~positive and a~negative instance of simple semi-unification.

\begin{example}
\label{xmp:ssu-yes}
Consider the simple constraints $\calC = \{ \config{0}{\alpha}{\epsilon} \doteq \config{\epsilon}{\beta}{1} \}$.
A possible substitution triple $(\varphi, \psi_0, \psi_1)$ which models $\config{0}{\alpha}{\epsilon} \doteq \config{\epsilon}{\beta}{1}$ is such that
$\varphi(\alpha) = \alpha$, ${\varphi(\beta) = \beta_1 \to \beta_2}$, and $\psi_0(\alpha) = \beta_2$.
Indeed, we have $\beta_1 \to \psi_0(\varphi(\alpha)) = \beta_1 \to \beta_2 = \varphi(\beta)$.
In fact, $\calC$ is related to the uniformly bounded two-stack machine $\SM = [(\config{0}{p}{\epsilon} \to \config{\epsilon}{q}{1}), (\config{\epsilon}{q}{1} \to \config{0}{p}{\epsilon})]$ from Example~\ref{xmp:sm-ub-nut}.
\end{example}

\begin{example}
\label{xmp:ssu-no}
Consider the simple constraints $\calC = \{ \config{0}{\alpha}{\epsilon} \doteq \config{\epsilon}{\alpha}{1} \}$.
There is no substitution triple $(\varphi, \psi_0, \psi_1)$ which models $\config{0}{\alpha}{\epsilon} \doteq \config{\epsilon}{\alpha}{1}$.
For some $\sigma$ such a~substitution triple would satisfy $\sigma \to \psi_0(\varphi(\alpha)) = \varphi(\alpha)$.
This is not possible because the syntax tree of $\sigma \to \psi_0(\varphi(\alpha))$ is strictly larger than the syntax tree of $\varphi(\alpha)$.
In fact, $\calC$ is related to the two-stack machine $\SM = [\config{0}{p}{\epsilon} \to \config{\epsilon}{p}{1}]$, which is not uniformly bounded, from Example~\ref{xmp:sm-nub-ut}.
\end{example}

The pivotal result in previous work~\cite[Section 4]{Dudenhefner20-SU} is a~many-one reduction from deterministic, simple two-stack machine uniform boundedness to simple semi-unification.
We recapitulate this result in the confluent case (Lemma~\ref{lem:cssm-ssu}).

\begin{lemma}
\label{lem:cssm-ssu}
Confluent, simple two-stack machine uniform boundedness (Problem~\ref{prb:cssmub}) many-one reduces to simple semi-unification (Problem~\ref{prb:ssu}).
\end{lemma}

\begin{proof}
We tacitly inject machine states into variables, i.e.\ $\bbS \subseteq \bbV$. 
Given a~confluent, simple two-stack machine $\SM$, we construct the set of simple constraints $\calC$~\cite[Definition~40]{Dudenhefner20-SU}:
\begin{align*}
\calC = \{ \config{a}{p}{\epsilon} \doteq \config{\epsilon}{q}{b} \mid\, & (\config{a}{p}{\epsilon} \to \config{\epsilon}{q}{b}) \in \SM \text{ or } 
(\config{\epsilon}{q}{b} \to \config{a}{p}{\epsilon}) \in \SM \}
\end{align*}
\begin{description}
\item[{\cite[Lemma~45]{Dudenhefner20-SU}}] If $\SM$ is uniformly bounded, then there exists a~substitution triple $(\varphi, \psi_0, \psi_1)$ which models each simple constraint in $\calC$.
\item[{\cite[Lemma~48]{Dudenhefner20-SU}}] If a~substitution triple $(\varphi, \psi_0, \psi_1)$ models each constraint in $\calC$, then the maximal depth of the syntax trees in the range of $\varphi$ induces a~uniform bound on the number of configurations reachable from any configuration in~$\SM$.\qedhere
\end{description}
\end{proof}

\begin{corollary}
Simple semi-unification (Problem~\ref{prb:ssu}) is undecidable.
\end{corollary}

\begin{remark}
Although the original construction~\cite[Lemma 45 and Lemma 48]{Dudenhefner20-SU} is based on determinism, only confluence is used in the actual proofs~\cite[Lemma~30 and Lemma~39]{Dudenhefner20-SU}.
While tedious to verify by hand, the available mechanization allows for simple replacement of determinism by confluence.
Meanwhile, the proof assistant guarantees correctness of (or indicates problems with) any related details.
This highlights the effectiveness of proof assistants to accommodate for changes in a~complex argument, reevaluating overall correctness.
\end{remark}

\subsection{Right/Left-Uniform, Two-inequality Semi-unification}
\label{subsec:r2su}
In this section, we consider a~restriction of semi-unification to only two inequalities with identical right-hand sides (Problem~\ref{prb:su2}).

\begin{problem}[Right-uniform, Two-inequality Semi-unification]
\label{prb:su2}
Given two inequalities $\sigma_0 \leq \tau$ and $\sigma_1 \leq \tau$ with identical right-hand sides, do there exist substitutions $\varphi, \psi_0, \psi_1$ such that 
$\psi_0(\varphi(\sigma_0)) = \varphi (\tau)$ and $\psi_1(\varphi(\sigma_1)) = \varphi (\tau)$?
\end{problem}

\begin{remark}
Simply put, the above Problem~\ref{prb:su2} can be stated as follows:
Given three terms $\sigma_0, \sigma_1, \tau$, are there substitutions $\varphi, \psi_0, \psi_1$ such that 
$\psi_0(\varphi(\sigma_0)) = \varphi (\tau) = \psi_1(\varphi(\sigma_1))$?
\end{remark}

We adjust the existing reduction from simple semi-unification to (non right-uniform) two-inequality semi-unification~\cite[Theorem 1]{Dudenhefner20-SU} to produce right-uniform inequalities.

\begin{lemma}
\label{lem:ssu-su2}
Simple semi-unification (Problem~\ref{prb:ssu}) many-one reduces to right-uniform, two-inequality semi-unification (Problem~\ref{prb:su2}).
\end{lemma}

\begin{proof}
Given constraints $\calC = \{\config{a_i}{\alpha_i}{\epsilon} \doteq \config{\epsilon}{\beta_i}{b_i} \mid i = 1, \ldots, n\}$, define $\tau = \beta_1 \to \cdots \to \beta_n$.
Let~$\gamma_i$ be fresh variables for $i=1, \ldots, n$ and define $\sigma^j = \sigma_1^j \to \cdots \to \sigma_n^j$ for $j \in \{0, 1\}$ where 
\[\sigma_i^j = \alpha_i \to \gamma_i \text{ if } a_i = j \text{ and } b_i = 0
\qquad
\sigma_i^j = \gamma_i \to \alpha_i \text{ if } a_i = j \text{ and } b_i = 1
\qquad
\sigma_i^j = \gamma_i  \text{ else}\]
We show that $\calC$ is solvable iff the right-uniform inequalities $\sigma^0 \leq \tau$ and $\sigma^1 \leq \tau$ are solvable.

\begin{itemize}
\item Assume that the substitution triple $(\varphi, \psi_0, \psi_1)$ models each simple constraint in~$\calC$.\\
W.l.o.g.\ $\varphi(\gamma_i) = \psi_0(\gamma_i) = \psi_1(\gamma_i) = \gamma_i$ for $i = 1, \ldots, n$.
By routine case analysis, we may adjust $\psi_0(\gamma_i)$ and $\psi_1(\gamma_i)$ for $i = 1, \ldots, n$ to obtain substitutions $\psi_0'$ and $\psi_1'$ such that $\psi_0'(\varphi(\sigma^0)) = \varphi (\tau) = \psi_1'(\varphi(\sigma^1))$.
\item Any solution $\varphi, \psi_0, \psi_1$ of $\sigma^0 \leq \tau$ and $\sigma^1 \leq \tau$ also models each constraint in~$\calC$.\qedhere
\end{itemize}

\end{proof}

\begin{corollary}
Right-uniform, two-inequality semi-unification (Problem~\ref{prb:su2}) is undecidable.
\end{corollary}

Akin to the right-uniform (identical right-hand sides) restriction, it is straightforward to formulate \emph{left-uniform} (identical left-hand sides), two-inequality semi-unification~\cite[Problem~34]{Dud21}, and prove its undecidability.
This is used as a~starting point in a~recent, alternative undecidability proof of System F typability~\cite[Lemma~40]{Dud21}.

\begin{problem}[Left-uniform, Two-inequality Semi-unification]
\label{prb:lu2su}
Given two inequalities $\sigma \leq \tau_0$ and $\sigma \leq \tau_1$ with identical left-hand sides, do there exist substitutions $\varphi, \psi_0, \psi_1$ such that 
$\psi_0(\varphi(\sigma)) = \varphi (\tau_0)$ and $\psi_1(\varphi(\sigma)) = \varphi (\tau_1)$?
\end{problem}

\begin{lemma}
\label{lem:2su-lu2su}
Right-uniform two-inequality semi-unification (Problem~\ref{prb:su2}) many-one reduces to left-uniform, two-inequality semi-unification (Problem~\ref{prb:lu2su}).
\end{lemma}

\begin{proof}
Given two right-uniform inequalities $\sigma_0 \leq \tau$ and $\sigma_1 \leq \tau$, let $\alpha_0, \alpha_1$ be fresh variables.
Construct the terms $\sigma' := \sigma_0 \to \sigma_1$, $\tau_0' := \tau \to \alpha_1$, and $\tau_1' := \alpha_0 \to \tau$.
We show that $\sigma_0 \leq \tau$ and $\sigma_1 \leq \tau$ are solvable iff the left-uniform inequalities $\sigma' \leq \tau_0'$ and $\sigma' \leq \tau_1'$ are solvable.

\begin{itemize}
\item Assume $\psi_0(\varphi(\sigma_0)) = \varphi (\tau)$ and $\psi_1(\varphi(\sigma_1)) = \varphi (\tau)$.
Construct the substitution~$\varphi'$ such that ${\varphi'(\alpha_0) = \psi_1(\varphi(\sigma_0))}$, ${\varphi'(\alpha_1) = \psi_0(\varphi(\sigma_1))}$, and otherwise $\varphi'(\alpha) = \varphi(\alpha)$.\\
We have ${\psi_0(\varphi'(\sigma')) = \varphi (\tau) \to \varphi'(\alpha_1) = \varphi' (\tau_0')}$ and
$\psi_1(\varphi'(\sigma')) = \varphi'(\alpha_0) \to \varphi (\tau) = \varphi' (\tau_1')$.
Therefore, $\varphi', \psi_0, \psi_1$ solve the left-uniform inequalities $\sigma' \leq \tau_0'$ and $\sigma' \leq \tau_1'$.
\item Assume ${\psi_0(\varphi(\underbrace{\sigma'}_{\sigma_0 \to \sigma_1})) = \varphi (\underbrace{\tau_0'}_{\tau \to \alpha_1})}$ and ${\psi_1(\varphi(\underbrace{\sigma'}_{\sigma_0 \to \sigma_1})) = \varphi (\underbrace{\tau_1'}_{\alpha_0 \to \tau})}$.
Immediately, we have ${\psi_0(\varphi(\sigma_0)) = \varphi (\tau)} = \psi_1(\varphi(\sigma_1))$.\qedhere
\end{itemize}
\end{proof}

\begin{corollary}
\label{lem:lu2su-undec}
Left-uniform, two-inequality semi-unification (Problem~\ref{prb:lu2su}) is undecidable.
\end{corollary}

\subsection{Main Result}
Finally, we compose the previously described reductions into a~comprehensive, constructive many-one reduction from Turing machine halting to semi-unification (Theorem~\ref{thm:tm-semiu}).
This constitutes the main result of the present work.

\begin{theorem}
\label{thm:tm-semiu}
Turing machine halting constructively many-one reduces to semi-unification (Problem~\ref{prb:semiu}).
\end{theorem}

\begin{proof}
By composition of Lemmas~\ref{lem:TM_to_CM2},~\ref{lem:mm-to-cm},~\ref{lem:cm1-dlsm},~\ref{lem:dlsm-cssm},~\ref{lem:cssm-ssu}, and~\ref{lem:ssu-su2}.
Constructivity of the argument is witnessed by an axiom-free mechanization~(Section~\ref{sec:mec}) using the Coq proof assistant.
\end{proof}

Since semi-unification is recursively enumerable (Remark~\ref{rem:semiu-rec}), it is \textsf{RE}-complete under many-one reductions (in the sense of~\cite[Chapter 7.2]{Rogers}).

\begin{corollary}
\label{cor:su-mc}
Semi-unification (Problem~\ref{prb:semiu}) is \textsf{RE}-complete under many-one reductions.
\end{corollary}

\section{Mechanization}
\label{sec:mec}
 
This section provides an overview over the constructive mechanization, using the Coq proof assistant~\cite{Coq}, of the many-one reduction from Turing machine halting to semi-unification.
The mechanization relies on, and is integrated into the growing Coq Library of Undecidability Proofs~\cite{CLUP20}. 
The reduction is axiom-free and spans approximately $20000$ lines of code, of which $3500$ is contributed by the present work.

At the core of the library is the following mechanized notion of many-one reducibility\footnote{\href{https://github.com/uds-psl/coq-library-undecidability/blob/coq-8.17/theories/Synthetic/Definitions.v}{\lstinline|theories/Synthetic/Definitions.v|} in the github repository of~\cite{CLUP20}}

\begin{lstlisting}[mathescape]
Definition reduction {X Y} (f: X -> Y) (P: X -> Prop) (Q: Y -> Prop) :=
  forall x, P x <-> Q (f x).

Definition reduces {X Y} (P: X -> Prop) (Q: Y -> Prop) :=
  exists f: X -> Y, reduction f P Q.

Notation "P $\preceq$ Q" := (reduces P Q).
\end{lstlisting}
In the above, a~predicate \lstinline|P| over the domain~\lstinline|X| many-one reduces to a~predicate \lstinline|Q| over the domain \lstinline|Y|, 
denoted \lstinline[mathescape]|P $\preceq$ Q|, %[
if there exists a~function \lstinline|f: X -> Y| such that for all \lstinline|x| in the domain~\lstinline|X| we have \lstinline|P x| iff \lstinline|Q (f x)|.
Implemented in axiom-free Coq any such function \lstinline|f: X -> Y| is computable, a~necessity oftentimes handled with less rigor in traditional (non-mechanized) proofs.
Additionally, an axiom-free proof of \lstinline|P x <-> Q (f x)| cannot rely on principles such as functional extensionality, the fan theorem, or the law of excluded middle.

Overall, the statement 
\lstinline[mathescape]|Theorem reduction : HaltTM 1 $\preceq$ SemiU|\footnote{\href{https://github.com/uds-psl/coq-library-undecidability/blob/coq-8.17/theories/SemiUnification/Reductions/HaltTM_1_to_SemiU.v}{\lstinline|theories/SemiUnification/Reductions/HaltTM_1_to_SemiU.v|}} %[
faithfully mechanizes the overall formal goal of a~constructive many-one reduction from Turing machine halting to semi-unification (Theorem~\ref{thm:tm-semiu}).
Correctness of the argument is witnessed by verification in axiom-free Coq, for which constructivity is certified using the \lstinline|Print Assumptions| command~\cite{CoqDoc}.
In fact, the particular many-one reduction function could be extracted from the proof of \lstinline[mathescape]|HaltTM 1 $\preceq$ SemiU| %[
as a~$\lambda$-term (in the call-by-value $\lambda$-calculus model of computation) using existing techniques~\cite{ForsterItp19}.

\newpage

In detail, key contributions of the present work are consolidated as part of the following conjunction of many-one reductions\footnote{\href{https://github.com/uds-psl/coq-library-undecidability/blob/coq-8.17/theories/SemiUnification/Reductions/HaltTM_1_chain_SemiU.v}{\lstinline|theories/SemiUnification/Reductions/HaltTM_1_chain_SemiU.v|}}.
\begin{lstlisting}[mathescape]
Theorem HaltTM_1_chain_SemiU :
  HaltTM 1 $\preceq$ iPCPb /\
  iPCPb $\preceq$ Halt_BSM /\
  Halt_BSM $\preceq$ MM2_HALTING /\
  MM2_HALTING $\preceq$ CM1_HALT /\
  CM1_HALT $\preceq$ SMNdl_UB /\
  SMNdl_UB $\preceq$ CSSM_UB /\
  CSSM_UB $\preceq$ SSemiU /\
  SSemiU $\preceq$ RU2SemiU /\
  RU2SemiU $\preceq$ SemiU.
\end{lstlisting}
The individual problems in the above chain of many-one reductions are as follows.
\begin{itemize}
\item \lstinline|HaltTM 1| is one-tape Turing machine halting, and native to the library as an initial undecidable problem, building upon prior work~\cite{ForsterTM20,AspertiR15} in computability theory;
\item \lstinline|iPCPb| is indexed, binary Post correspondence problem, mechanized in~\cite{ForsterPCP};
\item \lstinline|Halt_BSM| is binary stack machine halting, mechanized in~\cite{ForsterMM2};
\item \lstinline|MM2_HALTING| is two-counter machine halting (Problem~\ref{prb:cm2halt}), mechanized in~\cite{ForsterMM2};
\item \lstinline|CM1_HALT| is one-counter machine 1-halting (Problem~\ref{prb:cm1halt});
\item \lstinline|SMNdl_UB| is uniform boundedness of deterministic, length-preserving two-stack machines (Problem~\ref{prb:dlsmub});
\item \lstinline|CSSM_UB| is uniform boundedness of confluent, simple stack machines (Problem~\ref{prb:cssmub});
\item \lstinline|SSemiU| is simple semi-unification (Problem~\ref{prb:ssu}), mechanized in~\cite{Dudenhefner20-SU};
\item \lstinline|RU2SemiU| is right-uniform, two-inequality semi-unification (Problem~\ref{prb:su2});
\item \lstinline|SemiU| is semi-unification (Problem~\ref{prb:semiu}), mechanized in~\cite{Dudenhefner20-SU}.
\end{itemize}
In the remainder of this section we sketch the contributed mechanizations of \lstinline|CM1_HALT|, \lstinline|SMNdl_UB|, and \lstinline|CSSM_UB| together with corresponding many-one reductions.

\subsection{One-counter Machines}
A one-counter machine \lstinline|Cm1|\footnote{\href{https://github.com/uds-psl/coq-library-undecidability/blob/coq-8.17/theories/CounterMachines/CM1.v}{\lstinline|theories/CounterMachines/CM1.v|}} is mechanized as a~list of instructions \lstinline|Instruction|, which are pairs \lstinline|State * nat| of a~program index and a~counter modifier.
According to Definition~\ref{def:CM}, \lstinline|State| is set to \lstinline|nat|.
Since a~counter modifier is strictly positive, an instruction $(j, d)$ (Definition~\ref{def:CM}) is mechanized as the pair \lstinline|(j, d-1)|.
\begin{lstlisting}
Definition State := nat.

Record Config := mkConfig { state: State; value: nat }.

Definition Instruction := State * nat.

Definition Cm1 := list Instruction.
\end{lstlisting}
The function \lstinline|step| mechanizes the step function on configurations, which are records with two fields: \lstinline|state| (program index) and \lstinline|value| (counter value).
The function \lstinline|nth_error| retrieves the current instruction \lstinline|(p, n)|.
On failure, \lstinline|step| enters a~trivial loop, mechanized as \lstinline|halting|.

1-halting (Problem~\ref{prb:cm1halt}) is mechanized as \lstinline|CM1_HALT|, that is the existence of a~number of steps \lstinline|n| after which the iterated step function reaches a~halting configuration starting from the configuration \lstinline${| state := 0; value := 1 |}$.
The counter modifier of one-counter machines \lstinline|M| in \lstinline|CM1_HALT| is less than $4$, mechanized as \lstinline|Forall (fun '(_, n) => n < 4) M|.

\begin{lstlisting}
Definition step (M: Cm1) (x: Config) : Config :=
  match (value x), (nth_error M (state x)) with
  | 0, _ => x (* halting configuration *)
  | _, None => x (* halting configuration *)
  | _, Some (p, n) => 
      match modulo (value x) (n+1) with
      | 0 => {| state := p; value := ((value x) * (n+2)) / (n+1) |}
      | _ => {| state := 1 + state x; value := value x |}
      end
  end.

Definition halting (M: Cm1) (x: Config) := step M x = x.

Definition CM1_HALT : {M: Cm1 | Forall (fun '(_, n) => n < 4) M} -> Prop :=
  fun '(exist _ M _) => 
    exists n, halting M (Nat.iter n (step M) {| state := 0; value := 1 |}).
\end{lstlisting}

The proof of \lstinline[mathescape]|MM2_HALTING $\preceq$ CM1_HALT|\footnote{\href{https://github.com/uds-psl/coq-library-undecidability/blob/coq-8.17/theories/CounterMachines/Reductions/MM2_HALTING_to_CM1_HALT.v}{\lstinline|theories/CounterMachines/Reductions/MM2_HALTING_to_CM1_HALT.v|}} %[
(Lemma~\ref{lem:mm-to-cm}) is by straightforward simulation and spans in total approximately 400~lines of code.

\subsection{Deterministic, Length-preserving Two-stack Machines} A two-stack machine \lstinline|SMN|\footnote{\href{https://github.com/uds-psl/coq-library-undecidability/blob/coq-8.17/theories/StackMachines/SMN.v}{\lstinline|theories/StackMachines/SMN.v|}} is mechanized as a~list of instructions \lstinline|Instruction| containing transition information.
The step relation \lstinline|step| on configurations \lstinline|Config| holds for instructions \lstinline|((r, s, x), (r', s', y))| such that \lstinline|x| is the current state, \lstinline|y| is the next state, \lstinline|r|, \lstinline|s| are the prefixes of the current respective left and right stacks, and \lstinline|r'|, \lstinline|s'| are the prefixes of the next respective left and right stacks.
The state space \lstinline|State| is set to \lstinline|nat| (we could have chosen any effectively enumerable type with decidable equality).
Configuration reachability \lstinline|reachable| is mechanized as the reflexive, transitive closure of \lstinline|step|.

\begin{lstlisting}
Definition State := nat.

Definition Symbol := bool.

Definition Stack := list Symbol.

Definition Config := Stack * Stack * State. 

Definition Instruction := Config * Config.

Definition SMN := list Instruction. 

Inductive step (M: SMN) : Config -> Config -> Prop :=
  | transition (v w r s r' s': Stack) (x y: State) : 
    In ((r, s, x), (r', s', y)) M -> 
    step M (r ++ v, s ++ w, x) (r' ++ v, s' ++ w, y).

Definition reachable (M: SMN) : Config -> Config -> Prop :=
  clos_refl_trans Config (step M).
\end{lstlisting}

A deterministic, length-preserving two-stack machine is mechanized as a~dependent pair
\lstinline${M : SMN | deterministic M /\ length_preserving M}$ of a~two-stack machine~\lstinline|M| and a~pair of (irrelevant) proofs showing that \lstinline|M| is deterministic and length-preserving.
Uniform boundedness (Problem~\ref{prb:dlsmub}) is mechanized as \lstinline|SMNdl_UB|, that is the existence of maximal length \lstinline|n| of exhaustive lists \lstinline|L| of reachable configurations \lstinline|Y| from any given configuration \lstinline|X|.

\begin{lstlisting}
Definition bounded (M: SMN) (n: nat) : Prop := 
  forall (X: Config), exists (L: list Config),
    (forall (Y: Config), reachable M X Y -> In Y L) /\ length L <= n.

Definition length_preserving (M: SMN) : Prop :=
  forall s t X s' t' Y, In ((s, t, X), (s', t', Y)) M ->
    length (s ++ t) = length (s' ++ t') /\ 1 <= length (s ++ t).

Definition SMNdl_UB :
  {M : SMN | deterministic M /\ length_preserving M} -> Prop :=
    fun '(exist _ M _) => exists (n: nat), bounded M n.
\end{lstlisting}

The proof of \lstinline[mathescape]|CM1_HALT $\preceq$ SMNdl_UB|\footnote{\href{https://github.com/uds-psl/coq-library-undecidability/blob/coq-8.17/theories/StackMachines/Reductions/CM1_HALT_to_SMNdl_UB.v}{\lstinline|theories/StackMachines/Reductions/CM1_HALT_to_SMNdl_UB.v|}} (Lemma~\ref{lem:cm1-dlsm}) %[
relies on a~variant of Hooper's argument~\cite{Hooper66}.
The particular mechanization details span approximately $2300$ lines of code, two thirds of which mechanize Hooper's argument\footnote{\href{https://github.com/uds-psl/coq-library-undecidability/blob/coq-8.17/theories/StackMachines/Reductions/CM1_HALT_to_SMNdl_UB/CM1_to_SMX.v}{\lstinline|theories/StackMachines/Reductions/CM1_HALT_to_SMNdl_UB/CM1_to_SMX.v|}}.
As a~side note, the mechanization does \emph{not} construct a~\enquote{mirror} machine~\cite[Part~IV]{Hooper66} for symbol search on the left stack (symmetric to symbol search on the right stack).
Instead, machine instructions may swap stacks.
This simplifies the construction and does not change uniform boundedness.
The particular machine \lstinline|M| consists of instructions simulating divisibility tests \lstinline|index_ops|, counter increase \lstinline|increase_ops|, stack traversal \lstinline|goto_ops|, and nested simulation initialization \lstinline|bound_ops|.
\begin{lstlisting}
Definition M : SMX :=
  locked index_ops ++ locked increase_ops ++
  locked goto_ops ++ locked bound_ops.
\end{lstlisting}

Since the proof structure for uniform bound verification is mostly by induction, extensive case analysis, and basic arithmetic, the mechanization benefits greatly from proof automation, i.e.\ Coq's \lstinline|lia|, \lstinline|nia|, and \lstinline|eauto| tactics~\cite{CoqDoc}.
For example, the proof automation tactic \lstinline|eauto with M| is used to solve machine specification correctness obligations.
It uses type-directed backwards search and relies on a~hint database of instruction specifications for the constructed machine \lstinline|M|.

To the best of the author's knowledge, the provided mechanization is the first that implements (a variant of) Hooper's argument.

\subsection{Confluent Two-stack Machines}
A confluent, simple two-stack machine \lstinline|cssm|\footnote{\href{https://github.com/uds-psl/coq-library-undecidability/blob/coq-8.17/theories/StackMachines/SSM.v}{\lstinline|theories/StackMachines/SSM.v|}} is mechanized as a~dependent pair \lstinline${M : ssm | confluent M}$ of a~simple two-stack machine~\lstinline|M| and an (irrelevant) proof that \lstinline|M| is confluent.
\begin{lstlisting}
Definition confluent (M: ssm) := forall (X Y1 Y2: config),
  reachable M X Y1 -> reachable M X Y2 -> 
    exists (Z: config), reachable M Y1 Z /\ reachable M Y2 Z.

Definition cssm := {M: ssm | confluent M}.
\end{lstlisting}

Uniform boundedness (Problem~\ref{prb:cssmub}) is mechanized as \lstinline|CSSM_UB|, that is the existence of maximal length \lstinline|n| of exhaustive lists \lstinline|L| of reachable configurations \lstinline|Y| from any given configuration \lstinline|X|.

\begin{lstlisting}
Definition bounded (M: ssm) (n: nat) := 
  forall (X: config), exists (L: list config),
    (forall (Y: config), reachable M X Y -> In Y L) /\ length L <= n.

Definition CSSM_UB (M: cssm) := exists (n: nat), bounded (proj1_sig M) n.
\end{lstlisting}

The proof of \lstinline[mathescape]|SMNdl_UB $\preceq$ CSSM_UB|\footnote{\href{https://github.com/uds-psl/coq-library-undecidability/blob/coq-8.17/theories/StackMachines/Reductions/SMNdl_UB_to_CSSM_UB.v}{\lstinline|theories/StackMachines/Reductions/SMNdl_UB_to_CSSM_UB.v|}} (Lemma~\ref{lem:cssm-ssu}) %[
is by revertible (confluent), local lookahead and spans approximately 800~lines of code.

The proof of \lstinline[mathescape]|CSSM_UB $\preceq$ SSemiU|\footnote{\href{https://github.com/uds-psl/coq-library-undecidability/blob/coq-8.17/theories/SemiUnification/Reductions/CSSM_UB_to_SSemiU.v}{\lstinline|theories/SemiUnification/Reductions/CSSM_UB_to_SSemiU.v|}} (Lemma~\ref{lem:cssm-ssu}) %[
is an almost verbatim copy of the corresponding proof of \lstinline[mathescape]|DSSM_UB $\preceq$ SSemiU| %[
from prior work~\cite[Section 5]{Dudenhefner20-SU}, in which determinism is replaced by confluence.
Most importantly, the key function $\zeta$~\cite[Definition~41]{Dudenhefner20-SU} is used to directly construct simple semi-unification solutions.

\section{Conclusion}
\label{sec:concl}
This work gives a~constructive many-one reduction from Turing machine halting to semi-unification (Theorem~\ref{thm:tm-semiu}).
It improves upon existing work~\cite{KTU93SemiU,Dudenhefner20-SU} regarding the following aspects.

First, previous approaches use Turing reductions to establish undecidability.
Therefore, such arguments are unable to establish \textsf{RE}-completeness under many-one reductions of semi-unification, shown in the present work (Corollary~\ref{cor:su-mc}).

Second, previous work relies on the undecidability of Turing machine immortality, which is not recursively enumerable, and obscures the overall picture.
In the present work, we avoid Turing machine immortality by adapting Hooper's ingenious construction~\cite{Hooper66} (also adapted in~\cite{KariO08}) to uniform boundedness (Lemma~\ref{lem:cm1-dlsm}).

Third, correctness of the reduction function is proven constructively (in the sense of axiom-free Coq), whereas previous work uses the principle of excluded middle, König's lemma~\cite{KTU93SemiU}, or the fan theorem~\cite{Dudenhefner20-SU}.
As a~result, anti-classical theories, such as synthetic computability theory~\cite{Bauer06}, may accommodate the presented results.

Fourth, computability of the many-one reduction function from Turing machines to semi-unification instances is established rigorously by its mechanization in the Coq proof assistant.
Traditionally, this aspect is treated less formally.

Fifth, the reduction is mechanized and contributed to the Coq Library of Undecidability Proofs~\cite{CLUP20}, building upon existing infrastructure.
Arguably, a~comprehensive mechanization is the \emph{only} feasible approach to verify a~reduction from Turing machine halting to semi-unification with high confidence in full detail.
The provided mechanization integrates existing work~\cite{Dudenhefner20-SU} into the library, and contributes a~first of its kind mechanized variant of Hooper's construction for uniformly bounded symbol search.

While this document provides a~high-level overview over the overall argument, surveyability (both local and global in the sense of~\cite{Bassler06}) is established mechanically.
Local surveyability is supported by the modular nature of the Coq Library of Undecidability Proofs.
That is, the mechanization of each reduction step can be understood and verified independently.
Global surveyability is supported by \lstinline|Theorem HaltTM_1_chain_SemiU| and the statement \lstinline[mathescape]|HaltTM 1 $\preceq$ SemiU| %[
(cf.\ Section~\ref{sec:mec}).
That is, the individually mechanized reduction steps do compose transitively.

The provided mechanization shows the maturity of the Coq proof assistant for machine-assisted verification of technically challenging proofs.
Neither Hooper's exact immortality construction~\cite{Hooper66} nor the exact semi-unification construction by Kfoury, Tiuryn, and Urzyczyn~\cite{KTU93SemiU} was mechanized.
Rather, the overall structure of the proof was revised to be mechanization-friendly.
For example, the simplicity and uniformity of one-counter machines as an intermediate model of computation serves exactly this purpose.

Already, building upon the present work, there is a~novel mechanization showing the undecidability of System F typability and type checking~\cite{Dud21}.
In addition, we envision further mechanized results.
For one, the undecidability of synchronous distributivity~\cite{AnantharamanELNR12} relies on uniform boundedness of semi-Thue systems that can be described as the presented (and mechanized) two-stack machines.
Further, since the underlying construction is already implemented, it is reasonable to mechanize a~many-one reduction from Turing machine halting to Turing machine immortality.
This would pave the way for further mechanized results.
For example, the undecidability of the finite variant property~\cite[Section 7]{BouchardGLN13} as well as several tiling problems~\cite{Kari07} rely on (variants of) Turing machine immortality.

\bibliographystyle{alphaurl}
\bibliography{bibliography}

\end{document}